\documentclass[12pt]{article}
\usepackage[utf8]{inputenc}
\usepackage{setspace}
\usepackage{amsmath,amsthm,amssymb,mathtools,bm}
\usepackage[margin=1in]{geometry}
\usepackage{booktabs,array,tabularx}
\newcolumntype{L}{>{\raggedright\arraybackslash}X}
\newcolumntype{C}{>{\centering\arraybackslash}X}
\usepackage{graphicx}
\usepackage{enumitem}
\usepackage{mathrsfs}
\usepackage[authoryear,round]{natbib}
\usepackage[hidelinks,hypertexnames=false]{hyperref}
\usepackage[section]{placeins}
\usepackage{float}

\allowdisplaybreaks
\newcommand{\E}{\mathbb{E}}

\newcommand{\one}{\mathbf{1}}

\newcommand{\tp}{^{\!\top}}
\newcommand{\Tt}{\mathcal{T}}
\newcommand{\Cc}{\mathcal{C}}

\newcommand{\tr}{\operatorname{tr}}
\DeclareMathOperator*{\argmin}{arg\,min}

\DeclareMathOperator{\diag}{diag}
\DeclareMathOperator{\dist}{dist}
\DeclareMathOperator{\rank}{rank}
\newcommand{\tablenote}[1]{\par\vspace{2pt}\begin{minipage}{0.98\textwidth}\footnotesize\textit{Note:} #1\end{minipage}}
\newcommand{\fignote}[1]{\par\vspace{1pt}\footnotesize\textit{Note:} #1}

\theoremstyle{plain}
\newtheorem{theorem}{Theorem}
\newtheorem{lemma}{Lemma}
\newtheorem{proposition}{Proposition}

\theoremstyle{definition}
\newtheorem{assumption}{Assumption}
\newtheorem{definition}{Definition}
\theoremstyle{remark}
\newtheorem{remark}{Remark}

\newcommand{\ActCondDonor}{$1.8$}
\newcommand{\ActCondTime}{$67.0$}
\newcommand{\ActDonorChanges}{$17$}
\newcommand{\ActFailDonor}{$0$}
\newcommand{\ActFailTime}{$0$}
\newcommand{\ActMaxKKTDonor}{$9.9\times10^{-19}$}
\newcommand{\ActMaxKKTTime}{$9.3\times10^{-9}$}
\newcommand{\ActMinInactiveDonor}{$5.2\times10^{-5}$}
\newcommand{\ActMinInactiveTime}{$8.6\times10^{-4}$}
\newcommand{\ActMinPositiveDonor}{$1.5\times10^{-3}$}
\newcommand{\ActMinPositiveTime}{$0.278$}
\newcommand{\ActNDonor}{$18$}
\newcommand{\ActNTime}{$1$}
\newcommand{\ActReplicates}{80}
\newcommand{\ActTimeChanges}{$1$}
\newcommand{\BiasAwareChi}{$2.43$}
\newcommand{\BiasAwareMeanAnchor}{$[-11.96,\,-0.72]$}

\newcommand{\BiasAwareSlopeAnchor}{$[-9.06,\,-2.81]$}

\newcommand{\CouplingBoundaryBlock}{$0.0012$}
\newcommand{\CouplingRadiusBlock}{$6.553$}
\newcommand{\CouplingRadiusCommon}{$6.544$}
\newcommand{\CouplingRtwoSetBlock}{$[0.254,\,0.996]$}
\newcommand{\CouplingVSetBlock}{$[18.5,\,97.4]$}
\newcommand{\CouplingWidthMax}{$1.021$}
\newcommand{\CouplingWidthMin}{$0.985$}

\newcommand{\DrawEntropyRank}{$4.72$}
\newcommand{\DrawPartRank}{$1.61$}
\newcommand{\DrawTopShare}{$0.62$}
\newcommand{\DrawTopThreeShare}{$0.77$}
\newcommand{\ExposureHilo}{$0.052$}
\newcommand{\ExposureMean}{$0.425$}
\newcommand{\ExposurePopMean}{$0.392$}

\newcommand{\ExposureSlope}{$0.042$}
\newcommand{\FullMetricSpread}{$20.2$}

\newcommand{\MCCovMCSE}{$0.005$}
\newcommand{\MCDraws}{$999$}
\newcommand{\MCFwerAppShaped}{$0.984$}
\newcommand{\MCFwerNearFace}{$0.944$}
\newcommand{\MCFwerPointMass}{$0.997$}
\newcommand{\MCFwerStable}{$0.950$}
\newcommand{\MCFwerZeroMargin}{$0.941$}

\newcommand{\MCRegionRegular}{$0.963$}
\newcommand{\MCReps}{$2{,}000$}

\newcommand{\MCSDRatioAppShaped}{$1.304$}
\newcommand{\MCSDRatioPointMass}{$1.611$}
\newcommand{\MCSDRatioRegular}{$1.052$ to $1.057$}

\newcommand{\MeanWidthPercent}{$30$}

\newcommand{\MedBoundaryDirPN}{0.0012}
\newcommand{\MedBreakdownChi}{$20.6$}

\newcommand{\MedDirDraws}{$800$}
\newcommand{\MedDirDrawsN}{800}
\newcommand{\MedHiloDirectionalCI}{$[-8.09,\,-6.55]$}
\newcommand{\MedHoldoutChi}{$2.43$}

\newcommand{\MedLowerVAnchored}{$19.8$}
\newcommand{\MedLowerVOne}{$31.6$}

\newcommand{\MedLowerVTwo}{$22.7$}

\newcommand{\MedLowerVexplAnchored}{$11.8$}

\newcommand{\MedMeanDirectionalCI}{$[-8.42,\,-4.26]$}
\newcommand{\MedMeanSEFull}{$0.843$}
\newcommand{\MedPopMeanDirectionalCI}{$[-7.99,\,-4.13]$}

\newcommand{\MedRtwoSet}{$[0.285,\,0.994]$}

\newcommand{\MedSlopeDirectionalCI}{$[-6.56,\,-5.30]$}
\newcommand{\MedSlopeSEFull}{$0.274$}
\newcommand{\MedStepCN}{0.10}
\newcommand{\MedStepGammaN}{0.8}

\newcommand{\MedVDirectionalSet}{$[18.9,\,92.9]$}

\newcommand{\MedVexplDirectionalSet}{$[10.5,\,74.2]$}

\newcommand{\PaperTitle}{Shared-Donor Inference for Fixed-Set Heterogeneity in Synthetic Difference-in-Differences}

\newcommand{\QuadScaleRatio}{$13$}
\newcommand{\RegionCond}{$29.2$}
\newcommand{\RegionDim}{$17$}
\newcommand{\RegionDimN}{17}
\newcommand{\RegionGaussRadius}{$5.25$}
\newcommand{\RegionRadius}{$6.54$}
\newcommand{\RegionRadiusRatio}{$1.25$}
\newcommand{\RegionRidge}{$1\times10^{-10}$}
\newcommand{\RegionRidgeMove}{$4.4\times10^{-10}$}
\newcommand{\RepFrameSize}{80}
\newcommand{\RepHalfMax}{$11.9$}
\newcommand{\RepLOOMax}{$2.9$}

\newcommand{\SDPluginAppShaped}{$1.271$}
\newcommand{\SDREffectiveRank}{$1.95$}
\newcommand{\SDRMinEig}{$1.2\times10^{-5}$}
\newcommand{\SDRPSDTargetChange}{$1.6\times10^{-20}$}
\newcommand{\SDRRank}{$17$}
\newcommand{\SDStepAppShaped}{$0.762$}
\newcommand{\SDStepRange}{$0.762$ to $1.043$}

\newcommand{\SlopeSensitiveDistance}{$0.546$}
\newcommand{\SlopeWidthPercent}{$0.4$}

\newcommand{\StudEntropyRank}{$7.34$}
\newcommand{\StudPartRank}{$2.07$}

\title{\PaperTitle}
\author{Takahiro Hoshino\thanks{Faculty of Economics, Keio University, Tokyo, Japan; RIKEN Center for Advanced Intelligence Project (AIP), Tokyo, Japan. E-mail: \texttt{bayesian@keio.jp}.} \and Makoto Nakakita\thanks{RIKEN Center for Advanced Intelligence Project (AIP), Tokyo, Japan.}}
\date{September 6, 2026}
\IfFileExists{results_macros.tex}{
\newcommand{\MedSlope}{$-5.93$}

\newcommand{\MedVSDRPoint}{$41.7$}

\newcommand{\MedVAB}{$44.7$}
\newcommand{\MedRhoMean}{$0.704$}
\newcommand{\MedRhoSlope}{$-1.555$}

\newcommand{\MedMeanSEDiag}{$0.458$}
\newcommand{\MedSlopeSEDiag}{$0.439$}

\newcommand{\MedHiloEst}{$-7.32$}

\newcommand{\MedRtwo}{0.826}

\newcommand{\MedRtwoAB}{0.786}
\newcommand{\MedMeanATT}{$-6.34$}
\newcommand{\MedVexplSDRPoint}{$35.1$}
\newcommand{\MedVexplAB}{$35.1$}

\newcommand{\MedPopMeanATT}{$-6.06$}

\newcommand{\MedRtwoDirectionalSet}{$[0.275,\,0.995]$}

\newcommand{\MedCleanRmspeMed}{$1.28$}
\newcommand{\MedCleanRmspeMax}{$3.40$}
\newcommand{\MedPseudoRmspeMed}{$1.29$}
\newcommand{\MedPseudoRmspeMax}{$3.61$}

\newcommand{\MedPrimaryTreated}{$18$}
\newcommand{\MedPrimaryDonors}{$11$}

}{}

\makeatletter
\def\@fnsymbol#1{\ensuremath{\number#1}}
\makeatother
\begin{document}
\maketitle

\begin{spacing}{1}
\begin{abstract}
Empirical studies often estimate several synthetic-control or synthetic
Difference-in-Differences effects from a common donor pool and then summarize
their heterogeneity.  Because the same donors enter several comparisons, the
first-stage errors are jointly distributed and affect different targets
differently.  We characterize three regimes for linear fixed-set summaries:
nonconcentration, root-rate donor exposure, and donor-negligibility.  We also
show that removing shared-donor uncertainty from a quadratic summary requires
the shared-donor covariance to vanish in every direction entering that summary.

We take a valid joint first-order experiment for the effect vector as the input
to the second stage.  It yields simultaneous intervals for a prespecified
linear family, a fixed-grid projection band, a centered effect-vector
confidence region, image sets for total and explained heterogeneity and their
ratio, and a sampling-center homogeneity test.  For the hard-simplex SDID
estimator, an operator-normalized numerical directional bootstrap re-estimates
the donor and time weights in every draw.  We use an exact critical-cone
derivative to check the numerical approximation.

In the Medicaid expansion application, omitting off-diagonal covariance
understates uncertainty for the equal-jurisdiction mean but overstates it for a
centered projection on the baseline uninsured rate.  Persistent counterfactual
mismatch is reported separately through deterministic sensitivity regions.
\end{abstract}
\noindent\textbf{Key Words:} synthetic control; donor sharing; treatment-effect
heterogeneity; variance components; fixed-set inference.
\end{spacing}
\newpage

\section{Introduction}\label{sec:intro}

Many empirical studies estimate a separate synthetic-control or synthetic
Difference-in-Differences (SDID) effect for each treated state, county,
hospital, school district, or firm
\citep{Abadie2010,Arkhangelsky2021,AbadieLHour2021,Robbins2017,Dube2015}.
They then ask where the estimated effects are largest, how much the effects
vary, and how that variation is related to observed group characteristics.
The usual second-stage analysis treats the group estimates as independent,
even when the same donors enter several first-stage comparisons.

When donors are reused, a shock to one donor can enter many estimated effects,
often through a low-rank component.  Its contribution to uncertainty is
therefore target-specific.  A common donor shock may dominate a level target,
such as a mean or an intercept, while contributing little to a centered slope.
It also enters the plug-in dispersion of the group estimates.  A naive
between-group variance then combines effect heterogeneity with first-stage
sampling noise, including the off-diagonal covariance induced by donor reuse.

A one-donor calculation makes the distinction explicit.  Suppose
\(\widehat\tau_g=\tau_g+u_g-\omega_gv\), where \(u_g\) is an own-group shock
and \(v\) is shared.  Then
\[
\widehat{\bar\tau}-\bar\tau=\bar u-\bar\omega v,
\]
whereas a centered slope on a standardized moderator \(\widetilde z_g\)
satisfies
\[
\widehat\beta-\beta
=
\frac{\sum_g\widetilde z_gu_g}{\sum_g\widetilde z_g^2}
-
\frac{\sum_g\widetilde z_g\omega_g}{\sum_g\widetilde z_g^2}v.
\]
The donor term is small when the centered moderator is nearly orthogonal to
the donor loadings.  By contrast, the naive effect dispersion contains
\(\operatorname{Var}_g(\omega_g)v^2
-2\operatorname{Cov}_g(\tau_g,\omega_g)v\).
We use this calculation to distinguish three questions: joint inference for
linear targets, first-stage adjustment for quadratic targets, and the
conditions under which a target retains a donor component.

\subsection{Contribution and scope}

We first characterize three concentration regimes under a prespecified
shared-donor experiment.  A target coefficient can have a nonvanishing donor
loading, a loading that remains in the root-rate law, or a loading that is
negligible at that rate.  These cases correspond, respectively, to failure of
concentration, consistency with a donor contribution in the limiting law, and
donor-negligibility.  For a positive-semidefinite quadratic target, the
complete donor polynomial disappears only under the stronger matrix condition
\(H\Sigma_DH=0\).

Our second contribution is a fixed-set inference procedure based on one joint
effect-vector law.  We use a common set of numerical effect draws to form
simultaneous intervals for the reported linear family and a band for the
projection on a fixed grid.  The same draws yield a confidence region for the
centered effect vector, image sets for total and explained heterogeneity and
their ratio, and a sampling-center homogeneity test.  In the hard-simplex SDID
application, every bootstrap draw re-solves the prespecified first-stage
programs.  The exact critical-cone derivative is used to assess the numerical
approximation.  Persistent counterfactual mismatch remains separate from the
sampling covariance.

\subsection{Relation to existing work}\label{sec:lit}

Synthetic control and SDID were developed primarily for one treated unit, a
pooled treatment effect, or a linear aggregate of counterfactual predictions
\citep{Abadie2010,Abadie2015,Arkhangelsky2021,Li2020,
CattaneoFengTitiunik2021,ChernozhukovWuthrichZhu2021}.  Penalized, augmented,
matrix-completion, generalized, and proximal methods provide alternative first
stages
\citep{AbadieLHour2021,BenMichael2021,Athey2021,Xu2017,
CuiTchetgen2024,ShiMiao2026}.  Our analysis takes the output of a prespecified first stage as its starting
point.  The comparison rule is not selected in this paper.

Existing methods for several treated synthetic controls allow dependence and
simultaneous linear prediction
\citep{CattaneoStaggered2025,CattaneoScpi2025}.  We instead study finite-set
functionals of the realized effect vector.  Total and explained heterogeneity
are quadratic functionals, and the loading matrix generated by donor reuse
determines which linear and quadratic targets retain donor uncertainty.

A second connection is to inference from estimated group effects, best linear
projections, and variance components
\citep{ChettyFriedmanRockoff2014,KSS2020,Bloom2017,
Armstrong2022,IgnatiadisWager2022,ChernozhukovGeneric2018,
SemenovaChernozhukov2021}.  Here, donor reuse and constrained weight
re-estimation generate the off-diagonal first-stage covariance.  In the
Medicaid application, including this covariance increases uncertainty for a
level mean and reduces it for a centered moderator slope.

Throughout the analysis, we fix the treated set, adoption date, donor pool,
comparison algorithm, tuning profile, scalar post-treatment contrast, and
reported target family.  We leave two distinct problems outside the scope of
this paper: selecting a comparison rule from generated aggregate outcomes, and
inference for dynamic group-by-event-time arrays, including minimax quadratic
risk and spectral structure.

\section{Finite-Set Targets and a Joint First-Stage Representation}\label{sec:setup}

\subsection{Targets}\label{sec:targets}

Let \(\Tt\) and \(\Cc\) denote the treated and donor groups, with sizes
\(G_1\) and \(G_0\).  Group \(g\) contains units \(i\), observed in
pre-treatment periods \(\mathcal T_0\) and post-treatment periods
\(\mathcal T_1\).  Unit covariates \(X_{git}\) enter the first stage.  Group
covariates \(W_g=(1,Z_g\tp)\tp\) are used only to summarize heterogeneity.

\begin{definition}[Group-specific ATT]\label{def:group-att}
The post-period group average treatment effect on the treated (ATT) is
\[
\tau_g=|\mathcal T_1|^{-1}\sum_{t\in\mathcal T_1}
\E\{Y_{git}(1)-Y_{git}(0)\mid g,t\},\qquad g\in\Tt.
\]
Individual and dynamic effects may vary; the second stage uses only this
post-period group average.
\end{definition}

Let \(Z\) and \(W\) stack \(Z_g\tp\) and \(W_g\tp\) over the treated groups.

\begin{definition}[Finite-set projection and variance decomposition]\label{def:targets}
Let
\[
M=I_{G_1}-G_1^{-1}\one\one\tp,
\qquad
Z_c=MZ,
\qquad
\Sigma_Z=G_1^{-1}Z_c\tp Z_c,
\]
and define
\[
H_Z=G_1^{-1}Z_c\Sigma_Z^{-1}Z_c\tp.
\]
The matrix \(H_Z\) is symmetric and idempotent, with \(H_Z\one=0\).  Assume
\(\rank(W)=p\), so the projection coefficient and explained variance are
uniquely defined.  Conditional on the realized treated groups and their
covariates, set
\begin{align}
\gamma^*&=\argmin_{\gamma}\sum_{g\in\Tt}(\tau_g-W_g\tp\gamma)^2,
\label{eq:projection-target}\\
V&=G_1^{-1}\tau\tp M\tau,
\qquad C=G_1^{-1}Z_c\tp\tau,\label{eq:variance-targets}\\
V_{\mathrm{expl}}&=C\tp\Sigma_Z^{-1}C,
\qquad R^2=V_{\mathrm{expl}}/V.\label{eq:r2-target}
\end{align}
The scalar \(R^2\) is undefined when \(V=0\).  If the centered effect-vector
confidence region contains the origin, we display \([0,1]\) as a boundary-safe
set rather than assigning a scalar value at the origin.  For \(w\) in the
empirical treated support,
\[
m_{\mathcal T}(w)=w^\top\gamma^*
\]
is the fitted finite-set projection.  Without additional assumptions, it is
neither a prediction for a newly drawn jurisdiction nor the effect of an
intervention on the moderator.
\end{definition}

Because the estimands condition on the realized treated groups, we refer to
them as finite-set parameters.  The term fixed-set asymptotics refers instead
to a repeated-sampling sequence in which \(G_1\) stays fixed while first-stage
information increases.  We do not impose an exchangeable or Gaussian law on
\(\tau_g\); any prespecified low-dimensional basis may be used for the
projection.

\subsection{Joint first-stage representation}\label{sec:factor}

We formulate the second-stage results in terms of a joint experiment for the
estimated effect vector, rather than tying them to one synthetic-control
algorithm.  For probability laws \(P\) and \(Q\) on a finite-dimensional
Euclidean space, write
\[
d_{\mathrm{BL}}(P,Q)
=
\sup_{f\in\mathrm{BL}_1}
\left|\int f\,dP-\int f\,dQ\right|,
\]
where \(\mathrm{BL}_1\) contains the functions bounded by one with Lipschitz
constant at most one.

\begin{assumption}[Joint effect-vector interface]
\label{as:joint-interface}
Fix \(G_1\).  Conditional on a design sigma-field \(\mathcal F_n\),
\[
\widehat\tau_n
=
\mu_n+a_n^{-1}Z_n+r_n,
\qquad
\mu_n=\tau_n+b_n,
\qquad
a_n\to\infty.
\]
Let
\[
m_n=E(Z_n\mid\mathcal F_n),
\qquad
\Omega_n=\operatorname{Var}(Z_n\mid\mathcal F_n).
\]
No mean-zero restriction is imposed on \(m_n\).  Let \(Z_n^*\) be a
conditional bootstrap draw, with
\[
\widehat\Omega_n=\operatorname{Var}^*(Z_n^*),
\qquad
\widehat\Sigma_n^\tau=a_n^{-2}\widehat\Omega_n.
\]
The following conditions hold.

\begin{enumerate}[label=(J\arabic*),leftmargin=2.8em]
\item For a conditionally tight \(Z\),
\[
d_{\mathrm{BL}}
\{\mathcal L(Z_n\mid\mathcal F_n),
\mathcal L(Z\mid\mathcal F_\infty)\}
\to_p0,
\]
and
\[
d_{\mathrm{BL}}
\{\mathcal L^*(Z_n^*),
\mathcal L(Z_n\mid\mathcal F_n)\}
\to_p0.
\]
Write
\(m=E(Z\mid\mathcal F_\infty)\) and
\(\Omega=\operatorname{Var}(Z\mid\mathcal F_\infty)\).

\item For the finite collection \(\mathscr L_n\) of reported linear maps,
including the target coefficients, fixed projection grid, and region metric,
\[
\max_{L\in\mathscr L_n}
\|L(\widehat\Omega_n-\Omega_n)L^\top\|_{\mathrm{op}}
=o_p(1).
\]

\item For the fixed collection
\(\mathscr H=\{H_1,\ldots,H_K\}\) of reported quadratic matrices,
\[
\max_k
|\operatorname{tr}\{H_k(\widehat\Omega_n-\Omega_n)\}|
=o_p(1),
\]
and, whenever regular quadratic covariance is reported,
\[
\max_{k,\ell}
|\mu_n^\top H_k(\widehat\Omega_n-\Omega_n)H_\ell\mu_n|
=o_p(1).
\]

\item For every \(L\in\mathscr L_n\), \(a_nLr_n=o_p(1)\).  For
\(H_k\in\mathscr H\), define
\[
\mathcal R_{H_k,n}
=
\frac{2}{G_1}(\mu_n+a_n^{-1}Z_n)^\top H_kr_n
+
\frac{1}{G_1}r_n^\top H_kr_n.
\]
Then \(a_n\mathcal R_{H_k,n}=o_p(1)\) in the regular regime and
\(a_n^2\mathcal R_{H_k,n}=o_p(1)\) in the local-boundary regime.
\end{enumerate}

When a variance-trace-corrected statistic is described as centered at the
\(a_n^{-1}\) scale, also require
\[
\max_k|\mu_n^\top H_km_n|=o_p(1).
\]
Without this condition, the limiting quadratic law retains the corresponding
first-order mean.
\end{assumption}

Here \(\mu_n=\tau_n+b_n\) is the sampling center.  Any
\(a_n^{-1}\)-order mean produced by the nonlinear first-stage map belongs to
the law of \(Z_n\), rather than to persistent counterfactual mismatch.  We
condition on the comparison rule and its tuning parameters before applying
Assumption~\ref{as:joint-interface}.  A different first stage can be used when
it supplies the same target-relevant joint experiment.

\subsection{Hard-simplex SDID implementation}
\label{sec:hard-simplex-implementation}

For the application, let \(\vartheta\) collect the state--year outcome moments
and define the complete group-level hard-simplex SDID map by
\[
\Phi(\vartheta)
=
\{\widehat\tau_g(\vartheta):g\in\mathcal T\}
\].
We solve one donor-weight program for each treated group and one common
time-weight program.  All of these programs are strongly convex quadratic
problems on probability simplices.  Supplementary Appendix A records the
objectives, KKT systems, and signed-cell coefficients.

On a rate coordinate \(p\in(0,1)\), we use the chart
\[
\mathcal P_{\delta,h}(p)
=
\operatorname{expit}
\left\{
\operatorname{logit}(p)+\delta\frac{h}{p(1-p)}
\right\};
\]
unrestricted coordinates use \(p+\delta h\).  This chart has derivative \(h\)
at zero and keeps every rate coordinate in \((0,1)\).

\begin{proposition}[Hard-simplex implementation of the joint experiment]
\label{prop:hard-simplex-interface}
Suppose the donor- and time-weight criteria satisfy Conditions S.A.1--S.A.4.

\begin{enumerate}[label=(\roman*),leftmargin=2.6em]
\item For fixed fitted weights, the effect-vector error has a signed-cell
representation.  Its shared-donor block can be written
\[
e=u-Dv,
\]
and the corresponding covariance retains all off-diagonal signed-cell terms.

\item The complete map \(\Phi\) is Hadamard directionally differentiable.  Its
derivative is obtained from the critical-cone quadratic programs and the
directional chain rule.

\item Let \(\widehat B_n\) be the actual-scale primitive replicate factor,
\(\widehat s_n=\|\widehat B_n\|_{\mathrm{op}}\),
\(\delta_n=c_\delta\widehat s_n^\gamma\), and
\(0<\gamma<1\).  Under the scale separation, chart, and optimization-error
conditions in Supplementary Appendix A, define
\[
U_n^*
=
\widehat s_n
\frac{
\widetilde\Phi_n\{\mathcal P_{\delta_n,\widehat B_n\xi/\widehat s_n}
(\widehat\vartheta_n)\}
-
\widetilde\Phi_n(\widehat\vartheta_n)
}{\delta_n}.
\]
Then
\[
d_{\mathrm{BL}}
\left[
\mathcal L^*(s_n^{-1}U_n^*),
\mathcal L\{\Phi'_{\vartheta_0}(G_\vartheta)\mid\mathcal F_\infty\}
\right]
\to_p0.
\]
\end{enumerate}
\end{proposition}

\begin{proof}
See Supplementary Appendix A.
\end{proof}

We use this numerical law for every primary repeated-sampling statement in the
Medicaid analysis.  The exact critical-cone derivative serves as a numerical
check.  Neither calculation removes persistent mismatch \(b\); Section
\ref{sec:weakhet} treats that component separately.

\section{Shared-Donor Inference for Fixed-Set Heterogeneity}
\label{sec:theory}

\subsection{Fixed-set reporting from a joint effect law}\label{sec:applicability}

The asymptotic sequence keeps \(G_1\) fixed and lets first-stage information
increase at rate \(a_n\).  A reported linear summary is
\(\theta_\ell=\ell^\top\mu\), estimated by
\(\widehat\theta_\ell=\ell^\top\widehat\tau\).  In the application, the
coefficient \(\ell\) defines the equal-jurisdiction mean, the prespecified
population-weighted mean, the centered moderator slope, or the 75th--25th
projected contrast.  We also evaluate the projected curve on a fixed grid.

For a fixed symmetric matrix \(H\), define
\[
Q_H(\mu)=G_1^{-1}\mu^\top H\mu,
\qquad
\widehat Q_{H,n}^{\mathrm{AN}}
=
G_1^{-1}\widehat\tau_n^\top H\widehat\tau_n
-
G_1^{-1}\operatorname{tr}(H\widehat\Sigma_n^\tau).
\]
With \(\Delta_n=\widehat\tau_n-\mu_n\),
\begin{equation}
\widehat Q_{H,n}^{\mathrm{AN}}-Q_H(\mu_n)
=
\frac{2}{G_1}\mu_n^\top H\Delta_n
+
\frac{1}{G_1}
\left\{
\Delta_n^\top H\Delta_n
-
\operatorname{tr}(H\widehat\Sigma_n^\tau)
\right\}.
\label{eq:exact-quadratic-expansion}
\end{equation}
This identity distinguishes regular quadratic inference from the local
sampling-center boundary.

\begin{theorem}[Fixed-set reporting from one joint effect-vector law]
\label{thm:fixedset-transfer}
Suppose Assumption~\ref{as:joint-interface} holds.

\begin{enumerate}[label=(\roman*),leftmargin=2.6em]

\item
Let \(L_n\in\mathscr L_n\) satisfy
\[
\|L_n-L\|_{\mathrm{op}}\to0.
\]
Then
\[
a_nL_n(\widehat\tau_n-\mu_n)
\rightsquigarrow
LZ,
\qquad
L_nZ_n^*
\rightsquigarrow_{\mathbb P}
LZ.
\]

Let
\[
\mathcal A_n
=
\{\ell_{1,n},\ldots,\ell_{J,n}\}
\]
be a prespecified finite family, with fixed \(J\) and limiting variances
bounded away from zero.  Define
\[
\widehat s_{j,n}^2
=
\operatorname{Var}^*
(\ell_{j,n}^\top Z_n^*)
\]
and let \(c_{1-\alpha,n}^{\max}\) be the conditional
\((1-\alpha)\)-quantile of
\[
\max_{1\le j\le J}
\left|
\frac{\ell_{j,n}^\top Z_n^*}{\widehat s_{j,n}}
\right|.
\]
If the limiting maximum has a continuous distribution at its
\((1-\alpha)\)-quantile, the intervals
\[
I_{j,n}
=
\left[
\ell_{j,n}^\top\widehat\tau_n
-
a_n^{-1}c_{1-\alpha,n}^{\max}\widehat s_{j,n},
\quad
\ell_{j,n}^\top\widehat\tau_n
+
a_n^{-1}c_{1-\alpha,n}^{\max}\widehat s_{j,n}
\right]
\]
satisfy
\[
P\left\{
\ell_{j,n}^\top\mu_n\in I_{j,n}
\text{ for every }j
\right\}
\to1-\alpha.
\]
A fixed projected-curve grid is included by adding its coefficient vectors to
\(\mathcal A_n\).

\item
Let \(H_1,\ldots,H_K\in\mathscr H\) be the reported quadratic matrices and
suppose \(\mu_n\to\mu\).  Then
\[
a_n
\left(
\widehat Q_{H_k,n}^{\mathrm{AN}}
-
Q_{H_k}(\mu_n)
\right)_{k=1}^K
\rightsquigarrow
\left(
\frac{2}{G_1}\mu^\top H_kZ
\right)_{k=1}^K.
\]
If, conditionally on \(\mathcal F_\infty\),
\[
Z\sim N(m,\Omega),
\]
the limiting mean and covariance are
\[
\beta_k
=
\frac{2}{G_1}\mu^\top H_km,
\qquad
\Gamma_{k\ell}
=
\frac{4}{G_1^2}
\mu^\top H_k\Omega H_\ell\mu.
\]
Thus variance-trace subtraction does not remove a nonzero first-order mean.

If \(H\in\mathscr H\) is symmetric and idempotent and
\[
a_nH\mu_n\to h,
\]
then
\[
a_n^2\widehat Q_{H,n}^{\mathrm{AN}}
\rightsquigarrow
\frac1{G_1}
\left\{
\|h\|_2^2
+
2h^\top Z
+
Z^\top HZ
-
\operatorname{tr}(H\Omega)
\right\}.
\]
The sampling-center boundary is \(h=0\).  The same display represents the
causal boundary \(H\tau_n=0\) only when
\[
a_nHb_n\to0.
\]

\item
Let
\[
M=I_{G_1}-G_1^{-1}\mathbf1\mathbf1^\top,
\qquad
\eta_n=M\mu_n,
\qquad
\widehat\eta_n=M\widehat\tau_n.
\]
Let \(S_n\) be data measurable and one-to-one on
\(\operatorname{range}(M)\), and suppose it converges on that subspace to an
\(\mathcal F_\infty\)-measurable limit \(S_0\).
Let \(c_{1-\alpha,n}^{\eta}\) be the conditional
\((1-\alpha)\)-quantile of
\[
\|S_nMZ_n^*\|_2
\]
and define
\[
\mathcal C_{\eta,n}
=
\left\{
v\in\operatorname{range}(M):
\|S_na_n(\widehat\eta_n-v)\|_2
\le c_{1-\alpha,n}^{\eta}
\right\}.
\]
If the limiting norm distribution is continuous at its critical value, then
\[
P\{\eta_n\in\mathcal C_{\eta,n}\}
\to1-\alpha.
\]

For \(v\in\operatorname{range}(M)\), define
\[
V(v)=G_1^{-1}v^\top v,
\qquad
E(v)=G_1^{-1}v^\top H_Zv.
\]
The joint quadratic image is
\[
\mathcal C_{VE,n}
=
\left\{
\{V(v),E(v)\}:
v\in\mathcal C_{\eta,n}
\right\}.
\]
Then
\[
\liminf_{n\to\infty}
P\left[
\{V_\mu,V_{\mathrm{expl},\mu}\}
\in
\mathcal C_{VE,n}
\right]
\ge1-\alpha.
\]
Thus the two quadratic targets are covered jointly, not merely by two
unrelated marginal statements.

Put
\[
\underline V_n
=
\inf_{v\in\mathcal C_{\eta,n}}V(v).
\]
If \(\underline V_n>0\), define
\[
\mathcal C_{VER,n}
=
\left\{
\left(
V(v),
E(v),
\frac{E(v)}{V(v)}
\right):
v\in\mathcal C_{\eta,n}
\right\}.
\]
Whenever \(V_\mu>0\),
\[
\liminf_{n\to\infty}
P\left[
\left(
V_\mu,
V_{\mathrm{expl},\mu},
R_\mu^2
\right)
\in
\mathcal C_{VER,n}
\right]
\ge1-\alpha
\]
on sequences for which
\(P(\underline V_n>0)\to1\).
If \(\underline V_n=0\), the scalar ratio is not continuously identified by
this region and the boundary-safe displayed ratio set is \([0,1]\).
At \(V_\mu=0\), the scalar parameter \(R_\mu^2\) is undefined.

\item
For the sampling-center homogeneity null
\[
H_{0,\mu}:M\mu_n=0,
\]
define
\[
T_n^{\mathrm{hom}}
=
a_n^2G_1^{-1}
\widehat\tau_n^\top M\widehat\tau_n,
\qquad
T_n^{*,\mathrm{hom}}
=
G_1^{-1}Z_n^{*\top}MZ_n^*.
\]
If the conditional law of
\[
G_1^{-1}Z^\top MZ
\]
is continuous at its \((1-\alpha)\)-quantile, the ideal
conditional-quantile test that rejects for large
\(T_n^{\mathrm{hom}}\) has asymptotic size \(\alpha\).
A finite number \(B\) of draws gives a Monte Carlo implementation of this
ideal test; the plus-one value has resolution \(1/(B+1)\).
\end{enumerate}
\end{theorem}

\begin{proof}
See Supplementary Appendix B.
\end{proof}

\begin{remark}[Pointwise and local-uniform boundary validity]
\label{rem:boundary-scope}
Theorem~\ref{thm:fixedset-transfer}(ii) and (iv) apply along every fixed local
sequence with \(a_nH\mu_n\to h\).  Uniformity over compact sets of \(h\)
requires uniform versions of the joint effect-vector and numerical
approximations.  We do not claim uniformity over an additional sequence of
changing first-stage critical regions.  The near-face Monte Carlo designs
therefore assess finite-sample behavior rather than a larger double-uniformity
statement.
\end{remark}

All four parts of Theorem~\ref{thm:fixedset-transfer} use the same conditional
effect-vector law.  They cover the simultaneous linear family, the projection
on a fixed grid, the centered region and its quadratic images, and the raw
sampling-center homogeneity test.  Variance-trace correction is used only for
the quadratic point estimates.  As
Equation~\eqref{eq:exact-quadratic-expansion} makes clear, the linear term
determines the regular quadratic limit, whereas the linear and centered
quadratic terms are of the same order near the boundary.  The target in each
case is \(\mu=\tau+b\); persistent mismatch is considered in
Section~\ref{sec:weakhet}.

\subsection{Shared-donor loading geometry}

To study donor reuse directly, write the first-stage error as
\(u_n-D_nv_n\).  The term \(u_n\) is specific to the treated units, whereas
\(D_nv_n\) is shared through the donor-loading matrix.  Theorem
\ref{thm:donor-regimes} separates three cases: a fixed donor floor, a
root-rate donor term, and donor-negligibility.  The result is stated for the
prespecified own--donor experiment in
Equation~\eqref{eq:two-rate-decomposition}; it is not a minimax impossibility
claim over other estimators or additional sources of information.  The
numerical directional law has a different role: it propagates the complete
adaptive hard-simplex map.

\begin{theorem}
[Shared-donor concentration regimes in a prespecified two-rate experiment]
\label{thm:donor-regimes}

Fix $G_1$ and a scalar post-treatment contrast. Suppose that, for a
prespecified comparison rule,
\begin{equation}
\widehat\tau_n-\mu_n
=
a_{T,n}^{-1}Z_{T,n}
-
a_{D,n}^{-1}D_nZ_{D,n}
+r_n,
\label{eq:two-rate-decomposition}
\end{equation}
where $a_{T,n}\to\infty$, $a_{D,n}>0$,
$\ell_n^\top Z_{T,n}=O_p(1)$, $Z_{D,n}=O_p(1)$, and
$\ell_n^\top r_n=o_p(1)$.

\begin{enumerate}[label=(\roman*),leftmargin=2.6em]
\item
If
\[
a_{D,n}^{-1}
(D_n^\top\ell_n)^\top Z_{D,n}
=o_p(1),
\]
then
\[
\ell_n^\top\widehat\tau_n
-
\ell_n^\top\mu_n
=o_p(1).
\]

\item
Suppose $a_{D,n}\to a_D\in(0,\infty)$,
$D_n^\top\ell_n\to d$,
$Z_{D,n}\rightsquigarrow Z_D$, and
\[
\Pr\{|d^\top Z_D|>\varepsilon\}>0
\]
for some $\varepsilon>0$. Then
\[
\ell_n^\top(\widehat\tau_n-\mu_n)
\rightsquigarrow
-a_D^{-1}d^\top Z_D,
\]
so $\ell_n^\top\widehat\tau_n$ is not consistent for
$\ell_n^\top\mu_n$.

\item
Suppose additionally that
\[
(\ell_n^\top Z_{T,n},Z_{D,n})
\rightsquigarrow
(Z_\ell,Z_D),
\qquad
a_{T,n}\ell_n^\top r_n=o_p(1),
\]
and
\[
q_n
=
\frac{a_{T,n}}{a_{D,n}}D_n^\top\ell_n
\longrightarrow q.
\]
Then
\[
a_{T,n}\ell_n^\top
(\widehat\tau_n-\mu_n)
\rightsquigarrow
Z_\ell-q^\top Z_D.
\]
Thus $q\ne0$ gives a consistent target whose root-$a_{T,n}$ law retains the
donor component, whereas $q=0$ makes donor uncertainty negligible at that
scale.

\item
Consider the fixed-donor limit in which $D_n\to D$ and
\[
E(Z_D\mid\mathcal F_\infty)=0,
\qquad
\operatorname{Var}(Z_D\mid\mathcal F_\infty)=\Omega_D.
\]
Define
\[
\Sigma_D=D\Omega_DD^\top,
\qquad
\mathcal E_D=\ker(\Sigma_D),
\qquad
\mathcal S_D=\operatorname{range}(D\Omega_D^{1/2}).
\]
Then
\[
\mathcal E_D
=
\{\ell:\Omega_D^{1/2}D^\top\ell=0\}
=
\{\ell:\ell^\top DZ_D=0
\ \text{a.s. conditionally on }\mathcal F_\infty\},
\]
and $\mathcal E_D$ is the unique maximal linear subspace on which the limiting
donor component vanishes. Moreover,
$\mathcal E_D^\perp=\mathcal S_D$.

\item
Let
\[
M=I_{G_1}-G_1^{-1}\mathbf1\mathbf1^\top,
\qquad
\bar d=G_1^{-1}D^\top\mathbf1,
\qquad
D_c=MD.
\]
Then
\[
D=\mathbf1\bar d^\top+D_c.
\]
For the equal-weight mean coefficient
$\ell_{\mathrm{eq}}=G_1^{-1}\mathbf1$,
\[
D^\top\ell_{\mathrm{eq}}=\bar d,
\]
whereas every centered coefficient satisfying
$\mathbf1^\top\ell=0$ obeys
\[
D^\top\ell=D_c^\top\ell.
\]
Thus centering removes the purely common donor loading but not differential
donor loading.

\item
Let
\[
W_D=DZ_D,
\qquad
E(W_D\mid\mathcal F_\infty)=0,
\qquad
\operatorname{Var}
(W_D\mid\mathcal F_\infty)=\Sigma_D,
\]
and assume \(E(\|W_D\|^2\mid\mathcal F_\infty)<\infty\).
For symmetric \(H\succeq0\), define
\[
Q_H(v)=G_1^{-1}v^\top Hv.
\]
The donor contribution to the variance-trace-corrected quadratic is the
algebraic identity
\[
\begin{aligned}
&
Q_H(\mu-W_D)-Q_H(\mu)
-
G_1^{-1}\operatorname{tr}(H\Sigma_D)
\\
&\qquad=
-\frac2{G_1}\mu^\top HW_D
+
\frac1{G_1}
\left\{
W_D^\top HW_D
-
\operatorname{tr}(H\Sigma_D)
\right\}.
\end{aligned}
\]
The complete donor polynomial vanishes almost surely for every \(\mu\) if and
only if
\[
HW_D=0
\quad\text{almost surely conditionally on }\mathcal F_\infty,
\]
which is equivalent to
\[
H\Sigma_DH=0.
\]
Because \(H,\Sigma_D\succeq0\), this is also equivalent to
\[
\Sigma_D^{1/2}H\Sigma_D^{1/2}=0.
\]

If the conditional third moments of \(W_D\) vanish---in particular, if its
conditional law is centrally symmetric---the linear and centered-quadratic
terms are uncorrelated.  If
\[
W_D\mid\mathcal F_\infty
\sim N(0,\Sigma_D),
\]
their total conditional variance is
\[
\frac4{G_1^2}
\mu^\top H\Sigma_DH\mu
+
\frac2{G_1^2}
\operatorname{tr}(H\Sigma_DH\Sigma_D).
\]
\end{enumerate}
\end{theorem}

\begin{proof}
See Supplementary Appendix B.
\end{proof}

For a linear target, Theorem~\ref{thm:donor-regimes} separates concentration
from root-rate regularity through the loading \(D^\top\ell\).  Equivalently,
the coefficient is located relative to
\[
\mathcal E_D=\ker(\Sigma_D),
\qquad
\mathcal S_D=\operatorname{range}(D\Omega_D^{1/2}).
\]
Centering removes the pure common component \(\mathbf1\bar d^\top\), but not
the differential loading \(D_c\).  It therefore need not make a target donor
robust or invariant to a change in the donor pool.

Quadratic targets require more than a small loading for one linear
coefficient.  Their donor contribution contains both the interaction with the
sampling center and a centered quadratic term, and the complete donor
polynomial disappears uniformly only when \(H\Sigma_DH=0\).  In the Medicaid
analysis we report continuous donor exposures for the prespecified linear
targets and the two corresponding donor scales for total and explained
heterogeneity.  No reported result uses an outcome residualizer or a
Riesz-corrected first stage.

\subsection{Reported estimators and one inferential law}
\label{sec:reported-estimators}

We write a reported linear target and its estimator as
\[
\theta_\ell=\ell^\top\mu,
\qquad
\widehat\theta_\ell=\ell^\top\widehat\tau.
\]
Primary uncertainty for the finite linear family is given by the max-statistic
interval in Theorem~\ref{thm:fixedset-transfer}(i).  Applying the same
construction to a fixed empirical grid gives the projected-curve band.

For \(H=M\) and \(H=H_Z\), define the variance-trace-corrected point
estimators
\[
\widehat Q_H^{\mathrm{AN}}
=
G_1^{-1}\widehat\tau^\top H\widehat\tau
-
G_1^{-1}\operatorname{tr}(H\widehat\Sigma^\tau).
\]
Thus
\[
\widehat V_{\mathrm{AN}}
=
\widehat Q_M^{\mathrm{AN}},
\qquad
\widehat V_{\mathrm{expl,AN}}
=
\widehat Q_{H_Z}^{\mathrm{AN}}.
\]
These are the reported quadratic point estimates.  Their primary uncertainty
statements are the globally optimized images of the whole centered effect-vector region
from Theorem~\ref{thm:fixedset-transfer}(iii)--(iv), not Wald intervals around
the analytic points.

For the sampling-center homogeneity null
\[
H_{0,\mu}:M\mu=0,
\]
the sole primary test is the raw quadratic test in
Theorem~\ref{thm:fixedset-transfer}(iv).  The official SDR and exact
critical-cone calculations are benchmarks in the Supplement.  The
ratio \(R^2\) is defined only for \(V>0\); if the effect region contains the
origin, the boundary-safe displayed set is \([0,1]\).

The household A/B cross-product is an assumption-distinct diagnostic.  If the
two half-sample first stages have conditional means \(\mu_A,\mu_B\), then
\[
E\left[
G_1^{-1}\widehat\tau_A^\top H\widehat\tau_B
-
G_1^{-1}\operatorname{tr}(H\widehat\Sigma^{AB})
\mid\mathcal F
\right]
=
G_1^{-1}\mu_A^\top H\mu_B.
\]
It estimates \(Q_H(\mu)\) only under the same-center condition
\(\mu_A=\mu_B=\mu\); the general statement and proof are in Supplementary
Appendix~B.

\subsection{Causal targets and mismatch sensitivity}\label{sec:weakhet}

The preceding sampling results are centered at $\mu=\tau+b$.  They apply to
the causal-effect vector $\tau$ when the target-specific projection of $b$ is
negligible at the normalization used by the relevant theorem.

\paragraph{Holdout benchmark for the mismatch scale.}
Using the prespecified two-year pseudo-post split, we re-estimate the complete
first-stage procedure on the remaining pre-treatment years.  Let
\(\widehat b_g^{\mathrm{val}}\) be the absolute pseudo-post prediction error and
let \(d_g\) be the intercept-adjusted RMSPE from the primary first stage.  Define
\[
\widehat\chi_g^{\mathrm{val}}
=
\frac{
|\widehat b_g^{\mathrm{val}}|
}{d_g\vee\underline d},
\qquad
\widehat\chi_{\max}^{\mathrm{val}}
=
\max_{g\in\mathcal T}
\widehat\chi_g^{\mathrm{val}}.
\]
The application reports the upward-rounded value
\[
\widehat\chi_{\mathrm{anchor}}
=
\frac{
\left\lceil
100\widehat\chi_{\max}^{\mathrm{val}}
\right\rceil
}{100}
\]
as one holdout-benchmarked point on the deterministic sensitivity curve
\[
\mathcal B(\chi)
=
\{b:|b_g|\le\chi d_g\}.
\]
Neither \(\widehat\chi_{\max}^{\mathrm{val}}\) nor
\(\widehat\chi_{\mathrm{anchor}}\) is assigned a coverage probability.  The
holdout calculation fixes the scale of one displayed sensitivity analysis; it
does not transport validation error probabilistically to the post-treatment
counterfactual error.

\section{Monte Carlo Evidence}\label{sec:mc}

We evaluate the numerical directional procedure in four designs: a stable
first stage, a near-face configuration, an exact face boundary, and a design
shaped to resemble the application.  Each row of
Table~\ref{tab:mc-procedure} uses \MCReps\ Monte Carlo replications and
\MCDraws\ conditional draws.

\begin{table}[t]\centering\footnotesize
\caption{Monte Carlo performance of the procedure used in the application.}\label{tab:mc-procedure}
\begin{tabular}{lrrrr}
\toprule
Audited design & SD ratio & Simultaneous FWER & Region coverage & Failures\\
\midrule
Stable face & $1.052$ & $0.950$ & $0.963$ & $0$\\
Near face & $1.056$ & $0.944$ & $0.963$ & $0$\\
Zero margin & $1.057$ & $0.941$ & $0.963$ & $0$\\
Application shaped & $1.304$ & $0.984$ & $0.993$ & $0$\\
\bottomrule
\end{tabular}

\tablenote{Estimand: the four prespecified linear targets and the centered
effect-vector region.  Design: four simulation designs, whose active-face geometry is verified at the population panel, at \MCReps\ replications
with \MCDraws\ draws; the degenerate point-mass stress design is in
Supplementary Appendix~C.  Primary law: the numerical directional bootstrap of
Proposition~\ref{prop:hard-simplex-interface}(iii), re-solving the complete map
along a normalized primitive direction at the prespecified step.  Familywise
coverage is simultaneous over the four targets and is the calibration criterion;
the region is calibrated at $0.95$.  SD ratio is the average bootstrap standard
deviation over the Monte Carlo standard deviation of the equal-weight mean.
Failures counts optimization failures and natural-domain violations together.
Monte Carlo standard error on a coverage rate is \MCCovMCSE\ at this replication
count.}
\end{table}

Familywise coverage is \MCFwerStable\ in the stable design,
\MCFwerNearFace\ near an active-face boundary, and \MCFwerZeroMargin\ at the
exact boundary.  The corresponding Monte Carlo intervals include $0.95$.
Region coverage is \MCRegionRegular.  No optimization or natural-domain
failure occurs; the latter follows from the logit chart used in every draw.

In the application-shaped design, familywise coverage is
\MCFwerAppShaped and the SD ratio is \MCSDRatioAppShaped; the corresponding
ratio in the other three designs is \MCSDRatioRegular.  The resulting
marginal intervals are conservative.  Supplementary Appendix~C separates the
plug-in, numerical-step, and first-order contributions and reports the
point-mass configuration as a deliberately degenerate stress design.

Supplementary Appendix~E reports a separate fixed-donor concentration
calculation.  A null-space loading contracts at the own-noise rate, whereas a
fixed positive loading approaches its donor variance floor.  Intervals based on
the full variance retain nominal coverage in both cases; only the former
concentrate to a point.

\section{Medicaid Expansion}\label{sec:empirical}

\subsection{Policy design, data, and targets}
\label{sec:med-design}

We study implementation of the ACA Medicaid expansion on January 1, 2014.  The
primary treated set contains \MedPrimaryTreated\ jurisdictions, after excluding
seven jurisdictions with an earlier adult expansion or waiver.  The primary
D11 donor set is a prespecified core group of nonexpansion states.  D17 adds
Wisconsin and five states that adopted after the 2014--2019 outcome window and
is used only for sensitivity.  Supplementary Appendix~D lists the jurisdictions
and records the policy sources.

We use the 2008--2019 ACS one-year Public Use Microdata Sample.
For state \(k\) and year \(t\), let \(\mathcal I_{kt}\) contain adults aged
19--64 who satisfy the contemporaneous low-income definition.  We define the
state--year outcome by
\[
Y_{kt}
=
\frac{\sum_{i\in\mathcal I_{kt}}PWGTP_{ikt}
\mathbf1\{\text{uninsured}_{ikt}=1\}}
{\sum_{i\in\mathcal I_{kt}}PWGTP_{ikt}}.
\]
The estimand is a repeated-cross-section effect for this annual population,
not a fixed-cohort effect.  The pre-treatment years are 2008--2013 and the
post-treatment years are 2014--2019.

The jurisdiction is the second-stage unit.  We condition on the realized
treated set, the full-sample baseline moderator, the moderator quantiles, and
the prespecified population weights.  The ACS person and household sampling
design supplies the repeated-sampling uncertainty.  The 80 replicate columns
are directions in the released survey frame, not 80 independent
jurisdictions.  We distinguish the structural donor component in
Theorem~\ref{thm:donor-regimes} from survey error in donor state--year rates,
the response of the fitted weights to those errors, and persistent mismatch.

The intercept-adjusted pre-fit RMSPE is
\[
d_g
=
\left[
T_0^{-1}\sum_{t\in\mathcal T_0}
\left\{
(Y_{gt}-\bar Y_{g,0})
-
\widehat\omega_g^\top(Y_{Ct}-\bar Y_{C,0})
\right\}^2
\right]^{1/2}.
\]
Its median and maximum are \MedCleanRmspeMed\ and \MedCleanRmspeMax\ percentage
points.  The corresponding pseudo-treated donor values are
\MedPseudoRmspeMed\ and \MedPseudoRmspeMax.  All treated jurisdictions fall within the full pseudo-treated range.  We
therefore retain the policy-defined cohort without using fit to revise its
membership.

Let \(N_g\) be the average 2008--2013 ACS weighted population for jurisdiction
\(g\), let \(\pi_g=N_g/\sum_hN_h\), and let \(z_g\) be its standardized
2008--2013 mean uninsured rate.  The four prespecified coefficients are
\[
\ell_{\mathrm{eq}}=G_1^{-1}\mathbf1,
\qquad
\ell_{\mathrm{pop}}=\pi,
\qquad
\ell_{\mathrm{slope}}=\frac{z}{z^\top z},
\qquad
\ell_{75-25}=(z_{.75}-z_{.25})\ell_{\mathrm{slope}}.
\]
We compute the percentiles over the realized treated support.  The slope is a
finite-set effect-modification projection on that support; it neither
identifies the effect of changing the baseline uninsured rate nor predicts the
effect for a new jurisdiction.

\subsection{ACS law and linear summaries}
\label{sec:med-survey-law}\label{sec:med-first}

Let \(\widehat\vartheta\) stack the full-sample state--year rates and let
\(\widehat\vartheta^{(r)}\), \(r=1,\ldots,80\), be the corresponding official
replicate estimates.  We use
\[
B_{80}
=
\sqrt{\frac4{80}}
\begin{bmatrix}
\widehat\vartheta^{(1)}-\widehat\vartheta&\cdots&
\widehat\vartheta^{(80)}-\widehat\vartheta
\end{bmatrix},
\qquad
\widehat s=\|B_{80}\|_{\mathrm{op}},
\]
and draw \(\widetilde G_b^*=B_{80}\xi_b/\widehat s\), with
\(\xi_b\sim N(0,I_{80})\).  For
\(\delta=c\widehat s^\gamma\), the actual-scale effect draw is
\[
U_b^*
=
\widehat s
\frac{\Phi\{\mathcal P_{\delta,\widetilde G_b^*}
(\widehat\vartheta)\}-\Phi(\widehat\vartheta)}{\delta},
\]
where \(\mathcal P\) is the rate-scale logit chart.  Every donor- and
 time-weight program is re-solved.  The reported \MedDirDraws\ directions had
no solver failure; the chart rules out an out-of-domain rate.  The exact
critical-cone calculation uses 1,500 directions as a numerical check.  The
common-index cross-year coupling is a maintained construction; Supplementary
Appendix D reports the year-block alternative.

\begin{figure}[t]\centering
\includegraphics[width=.94\textwidth]{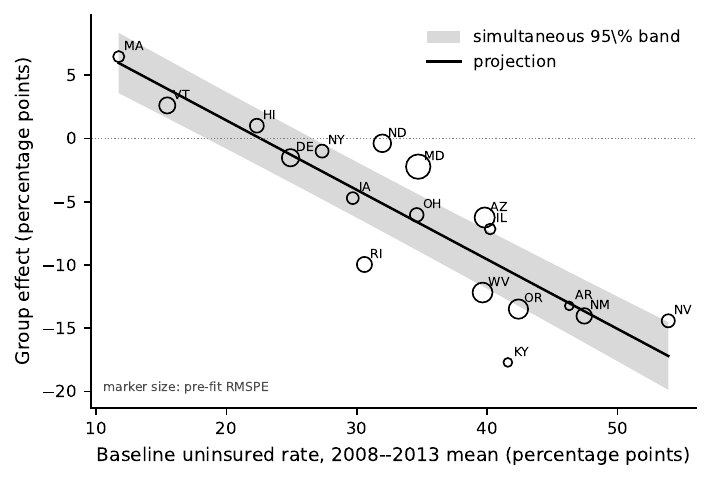}
\caption{State-specific Medicaid effects and their projection on baseline uninsured rates.}\label{fig:acs-bite}
\fignote{The figure reports the finite-set projection of the sampling-center effect
vector on the 2008--2013 mean uninsured rate over the realized
\MedPrimaryTreated-jurisdiction support.  Units are percentage points.  The band is simultaneous over the grid, by
Theorem~\ref{thm:fixedset-transfer}(i), and is centered at the sampling-center
projection.  Design: T\MedPrimaryTreated${}\times{}$D\MedPrimaryDonors.}
\end{figure}

\begin{table}[t]
\centering\footnotesize
\caption{Primary Medicaid linear summaries and target-specific donor exposure.}
\label{tab:med-linear-route1}
\begin{tabularx}{\textwidth}{Lcccc}
\toprule
Target
& Estimate
& Simultaneous 95\% CI
& Donor exposure
& Diagonal benchmark SE\\
\midrule
Equal-jurisdiction mean ATT (pp)
& \MedMeanATT
& \MedMeanDirectionalCI
& \ExposureMean
& \MedMeanSEDiag\\
Population-weighted mean ATT (pp)
& \MedPopMeanATT
& \MedPopMeanDirectionalCI
& \ExposurePopMean
& --\\
Centered baseline uninsured-rate slope (pp/SD)
& \MedSlope
& \MedSlopeDirectionalCI
& \ExposureSlope
& \MedSlopeSEDiag\\
75th--25th projected contrast (pp)
& \MedHiloEst
& \MedHiloDirectionalCI
& \ExposureHilo
& --\\
\bottomrule
\end{tabularx}
\tablenote{The four intervals use one studentized max-statistic critical value
and cover the corresponding sampling-center summaries simultaneously under
the maintained joint effect-vector law.  Donor
exposure is
\(\mathfrak g_D(\ell)=\{\ell^\top\widehat\Sigma_D\ell\}^{1/2}\), in
percentage points.  The last column reports standard errors after deleting off-diagonal
covariance and is shown only for the mean and slope.  A causal ATT
interpretation also requires the corresponding linear functional of persistent
mismatch to satisfy the reported sensitivity restriction.}
\end{table}

The full and diagonal standard errors for the equal mean are
\MedMeanSEFull\ and \MedMeanSEDiag\ percentage points.  For the centered slope
they are \MedSlopeSEFull\ and \MedSlopeSEDiag\ percentage points per standard
deviation.  Omitting off-diagonal covariance therefore understates uncertainty
for the mean and overstates it for the slope.  The signed off-diagonal
contributions relative to the full variance are \MedRhoMean\ and
\MedRhoSlope; the latter is not a variance-component share.

For the equal mean, donor exposure is \ExposureMean\ percentage points, and
the coefficient lies in the estimated donor-sensitive complement.  The
centered slope has exposure \ExposureSlope\ and relative donor-sensitive
distance \SlopeSensitiveDistance.  Thus centering removes the common loading
without necessarily removing the differential loading.  The four target
exposures are reported in Supplementary Appendix~D, with threshold sensitivity
in Supplementary Appendix~B.

\subsection{Quadratic heterogeneity and policy-design sensitivity}
\label{sec:med-trace}\label{sec:med-design-geometry}\label{sec:med-quadratic}

For \(H=M\) and \(H=H_Z\), we report
\[
\widehat Q_H^{\mathrm{AN}}
=
G_1^{-1}\widehat\tau^\top H\widehat\tau
-
G_1^{-1}\operatorname{tr}(H\widehat\Sigma^\tau).
\]
The variance-trace correction targets the sampling center
\(\mu=\tau+b\).  Under the regular centering condition it subtracts the
estimated covariance trace, but it does not remove persistent mismatch or a
nonzero first-order mean in the directional law.  We obtain the confidence
sets by optimizing each quadratic summary over the full centered effect-vector
region.  Table~\ref{tab:med-quadratic-route1} reports the point estimates,
these image sets, and the household A/B diagnostic in separate columns.

\begin{table}[t]
\centering\small
\caption{Variance-trace-corrected Medicaid heterogeneity summaries and confidence sets.}
\label{tab:med-quadratic-route1}
\begin{tabularx}{\textwidth}{>{\raggedright\arraybackslash}p{0.31\textwidth}CCC}
\toprule
Target
& Variance-trace-corrected point estimate
& Primary 95\% image set
& Household A/B check\\
\midrule
Total heterogeneity \(V_\mu\) (pp\(^2\))
& \MedVSDRPoint
& \MedVDirectionalSet
& \MedVAB\\
Explained heterogeneity \(V_{\mathrm{expl},\mu}\) (pp\(^2\))
& \MedVexplSDRPoint
& \MedVexplDirectionalSet
& \MedVexplAB\\
Explained share \(R_\mu^2\)
& \MedRtwo
& \MedRtwoDirectionalSet
& \MedRtwoAB\\
\bottomrule
\end{tabularx}
\tablenote{The point estimates subtract the estimated covariance trace from the full
sample quadratic.  The confidence sets are globally optimized images of the
entire numerical-directional centered effect-vector region, not Wald intervals
around those points.  The A/B quantities
use separately re-estimated household halves and are assumption-distinct
checks; they target the same quadratic object only under the same-center
condition in Supplementary Appendix~B.  The scalar \(R^2\) is defined only for \(V>0\); \([0,1]\) is displayed
when the global denominator margin is zero.}
\end{table}

\begin{table}[t]
\centering\small
\caption{Policy-design sensitivity and donor geometry.}
\label{tab:med-design-comparison}
\begin{tabular}{lrrrrrrr}
\toprule
Design & \multicolumn{2}{c}{Equal mean} & \multicolumn{2}{c}{Centered slope} & $\|\bar d\|$ & $\|D_c\|_F$ & $\dim\mathcal E_D$\\
\cmidrule(lr){2-3}\cmidrule(lr){4-5}
 & Estimate & Width & Estimate & Width & & &\\
\midrule
T18 x D11 & $-6.34$ & $4.16$ & $-5.93$ & $1.26$ & $0.320$ & $0.654$ & $7$\\
T18 x D17 & $-6.61$ & $3.21$ & $-5.91$ & $1.25$ & $0.263$ & $0.697$ & $3$\\
T25 x D17 & $-6.25$ & $2.40$ & $-5.86$ & $1.18$ & $0.254$ & $0.899$ & $8$\\
\bottomrule
\end{tabular}

\tablenote{Effect estimates and interval widths are in percentage points for
the equal mean and percentage points per standard deviation for the centered
slope.  Width is the length of the simultaneous 95\% interval.  T18-by-D11 is
the sole primary design; the other rows are prespecified design sensitivities.
The T18-by-D17 row changes only the donor pool.  The historical T25-by-D17
row changes both the cohort and the donor pool and is not used to isolate a
donor-pool effect.}
\end{table}

Holding the T18 cohort fixed, replacing D17 by D11 widens the equal-mean
interval by \MeanWidthPercent\ percent and changes the slope width by
\SlopeWidthPercent\ percent.  The comparison is empirical rather than a
consequence of invariance: a centered coefficient removes the common loading,
but the differential loading may change with the donor pool.

For a positive-semidefinite quadratic target, the two donor scales are
\[
\mathfrak a_H^2
=
\frac4{G_1^2}\mu^\top H\Sigma_DH\mu,
\qquad
\mathfrak k_H^2
=
\frac2{G_1^2}\operatorname{tr}(H\Sigma_DH\Sigma_D).
\]
Neither \(M\Sigma_DM\) nor \(H_Z\Sigma_DH_Z\) vanishes in the primary design.
For total heterogeneity the first scale is about \QuadScaleRatio\ times the
second.  Supplementary Appendix~D gives both scales and the
matrix-annihilation residual.  The primary region uses the
prespecified diagonal-standardization metric.  Supplementary Appendix~D reports
ridge and whitening sensitivities, and Supplementary Appendix~B gives the
global trust-region certificate.

\subsection{Boundary and mismatch sensitivity}
\label{sec:med-robust}

The sampling-center homogeneity test compares
\[
T_{\mathrm{obs}}
=G_1^{-1}\widehat\tau^\top M\widehat\tau
\quad\text{with}\quad
T_b^*=G_1^{-1}U_b^{*\top}MU_b^*.
\]
None of the \MedDirDraws\ directional values exceeds the observed statistic,
so the finite-run plus-one value is
\[
p_{\mathrm{boot}}
=
\MedBoundaryDirPN
=
1/(\MedDirDrawsN+1).
\]
This is the resolution of the conditional Monte Carlo calculation.  The null is
\(M\mu=0\); it is a causal homogeneity null only under the corresponding
centered-mismatch restriction.

For causal sensitivity, let
\[
\mathcal B(\chi)=\{b:|b_g|\le\chi d_g\}.
\]
If a sampling region covers \(\mu\) and \(b\in\mathcal B(\chi)\), its
Minkowski enlargement by \(-\mathcal B(\chi)\) covers \(\tau\).  In particular,
a linear interval expands by
\[
\chi\sum_{g=1}^{G_1}|\ell_g|d_g
\]
on each side.  The multipliers \(1\), \(2\), and the holdout-anchored value
\(\MedHoldoutChi\) index deterministic sensitivity sets, not confidence
levels.  At \(\chi=\BiasAwareChi\), the equal-mean and centered-slope intervals
are \BiasAwareMeanAnchor\ and \BiasAwareSlopeAnchor.  The lower values for
total heterogeneity are \MedLowerVOne\ and \MedLowerVTwo\ pp\(^2\) at
\(\chi=1\) and \(2\); at the holdout-anchored multiplier, the lower values for
total and explained heterogeneity are \MedLowerVAnchored\ and
\MedLowerVexplAnchored\ pp\(^2\).  The total-heterogeneity breakdown multiplier
is \MedBreakdownChi.  No validation-to-post coverage probability is assigned to
these sets.

The policy list and the remaining fit, leave-one-jurisdiction,
target-coefficient, cross-year-coupling, replicate-frame, and region-metric
checks appear in Supplementary Appendix D.

\section{Discussion}

We derived three concentration regimes for linear summaries of a prespecified
shared-donor effect vector.  A target may retain a fixed donor floor, lose that
floor while keeping a donor term in its root-rate law, or be donor-negligible
at that rate.  Centering eliminates the common loading, but generally not the
differential loadings.  For a positive-semidefinite quadratic target, uniform
donor robustness requires the stronger condition \(H\Sigma_DH=0\).

The Medicaid results make the target dependence visible.  Omitting
off-diagonal covariance reduces reported uncertainty for the
equal-jurisdiction mean and increases it for the centered projection on the
baseline uninsured rate.  Linear donor exposure is small for the centered
slope, while total and explained heterogeneity retain nonzero quadratic donor
scales.  No single treatment of the dependence is conservative for every
reported target.

One numerical directional law underlies all primary uncertainty statements.
We apply it to the simultaneous linear family, the projection on a fixed grid,
the centered region, the quadratic image sets, and the sampling-center
boundary test.  The hard-simplex calculation provides one implementation of
the joint interface.  A different first stage must supply its own joint law
and the target-relevant covariance conditions used here.

The repeated-sampling target is \(\mu=\tau+b\).  Neither covariance
estimation, variance-trace correction, nor the directional bootstrap removes
persistent counterfactual mismatch.  The RMSPE-scaled sets instead show how
the conclusions change under a stated deterministic restriction on that
mismatch.  Throughout, we condition on a fixed treated set, a common adoption
date, one scalar post-treatment contrast, and a prespecified comparison rule.
Selecting the rule from generated aggregate outcomes and analyzing dynamic
effect arrays are separate problems.

\clearpage
\appendix
\numberwithin{theorem}{section}
\numberwithin{lemma}{section}
\numberwithin{proposition}{section}
\numberwithin{corollary}{section}
\numberwithin{assumption}{section}
\numberwithin{definition}{section}
\numberwithin{remark}{section}
\numberwithin{equation}{section}
\numberwithin{table}{section}
\numberwithin{figure}{section}
\renewcommand{\thesection}{\Alph{section}}
\setcounter{section}{0}

\begin{center}
{\Large\bfseries Supplementary Material}
\end{center}
\medskip

\noindent Appendix~\ref{app:first-stage} describes the hard-simplex SDID map and
the ACS replicate construction, and Appendix~\ref{app:proofs} contains the
proofs and computational lemmas used in the main article.  The remaining
appendices report the procedure-matched numerical diagnostics
(Appendix~\ref{app:mc}), the Medicaid implementation and prespecified
sensitivities (Appendix~\ref{app:medicaid}), and a short fixed-donor
concentration calculation (Appendix~\ref{app:fixed-donor}).
\bigskip

\renewcommand{\thesection}{\Alph{section}}
\setcounter{section}{0}
\setcounter{table}{0}\renewcommand{\thetable}{S.\arabic{table}}
\setcounter{figure}{0}\renewcommand{\thefigure}{S.\arabic{figure}}
\numberwithin{equation}{section}
\numberwithin{theorem}{section}
\numberwithin{lemma}{section}
\numberwithin{proposition}{section}
\numberwithin{corollary}{section}
\numberwithin{assumption}{section}
\numberwithin{definition}{section}
\numberwithin{remark}{section}

\section{Hard-Simplex SDID Map and Numerical Directional Validation}
\label{app:first-stage}

This appendix specifies the group-level hard-simplex SDID estimator and the
ACS replicate construction used in the Medicaid analysis.  Its input is the
panel of design-weighted state--year uninsured rates.  We work directly with
that panel: the reported analysis neither residualizes the outcome nor uses a
Riesz or cross-fitted nuisance correction.

\subsection{Group-level estimator}

Let
\[
Y=(Y_{kt}):
k\in\mathcal T\cup\mathcal C,\quad
t\in\mathcal T_0\cup\mathcal T_1
\]
be the group--time panel, with \(T_0=|\mathcal T_0|\) and
\(T_1=|\mathcal T_1|\).  For a row vector \(y_{k,0}\) over the pre-period,
write
\[
\bar Y_{k,0}=T_0^{-1}\sum_{t\in\mathcal T_0}Y_{kt},
\qquad
\bar Y_{k,1}=T_1^{-1}\sum_{t\in\mathcal T_1}Y_{kt}.
\]

For treated group \(g\), define
\[
A_C(Y)
=
\left(
Y_{ht}-\bar Y_{h,0}
\right)_{t\in\mathcal T_0,\ h\in\mathcal C}
\in\mathbb R^{T_0\times G_0},
\]
and
\[
b_g(Y)
=
\left(
Y_{gt}-\bar Y_{g,0}
\right)_{t\in\mathcal T_0}
\in\mathbb R^{T_0}.
\]
The donor weights solve
\begin{equation}
\widehat\omega_g(Y)
=
\argmin_{\omega\in\Delta_{G_0}}
\left\{
\|A_C(Y)\omega-b_g(Y)\|_2^2
+
\rho_{\omega}(Y)\|\omega\|_2^2
\right\},
\label{eq:donor-program}
\end{equation}
where
\[
\Delta_m
=
\{x\in\mathbb R^m:x\ge0,\ \one^\top x=1\}.
\]

For the common time weights, let
\[
M_C(Y)
=
\left(
Y_{ht}
-
G_0^{-1}\sum_{j\in\mathcal C}Y_{jt}
\right)_{h\in\mathcal C,\ t\in\mathcal T_0}
\in\mathbb R^{G_0\times T_0},
\]
and
\[
y_C(Y)
=
\left(
\bar Y_{h,1}
-
G_0^{-1}\sum_{j\in\mathcal C}\bar Y_{j,1}
\right)_{h\in\mathcal C}.
\]
The time weights solve
\begin{equation}
\widehat\lambda(Y)
=
\argmin_{\lambda\in\Delta_{T_0}}
\left\{
\|M_C(Y)\lambda-y_C(Y)\|_2^2
+
\rho_{\lambda}(Y)\|\lambda\|_2^2
\right\}.
\label{eq:time-program}
\end{equation}

Define the signed time coefficient
\[
c_t(\lambda)
=
\begin{cases}
-\lambda_t,&t\in\mathcal T_0,\\[2mm]
T_1^{-1},&t\in\mathcal T_1,
\end{cases}
\]
and
\[
d_k(Y,\lambda)
=
\sum_t c_t(\lambda)Y_{kt}.
\]
The fitted group effect is
\begin{equation}
\widehat\tau_g(Y)
=
d_g\{Y,\widehat\lambda(Y)\}
-
\widehat\omega_g(Y)^\top
d_{\mathcal C}\{Y,\widehat\lambda(Y)\}.
\label{eq:reported-sdid-map}
\end{equation}

We fix the ridge functions \(\rho_\omega(Y)\) and \(\rho_\lambda(Y)\), the
pre- and post-treatment periods, and the number of projected-gradient
iterations before generating directional draws.  None of these choices is
re-selected within a draw.

\subsection{Signed-cell representation}

For \(k\in\{g\}\cup\mathcal C\), define
\[
a_{gk}(\omega)
=
\begin{cases}
1,&k=g,\\
-\omega_k,&k\in\mathcal C.
\end{cases}
\]

\begin{lemma}[Signed-cell representation]
\label{lem:signed-cell}
For every \((\omega,\lambda)\in
\Delta_{G_0}\times\Delta_{T_0}\),
\[
\psi_g(\omega,\lambda;Y)
=
\sum_{k\in\{g\}\cup\mathcal C}
\sum_{t\in\mathcal T_0\cup\mathcal T_1}
a_{gk}(\omega)c_t(\lambda)Y_{kt}.
\]
If \(\varepsilon\) stacks the group--time estimation errors and
\(\Omega_Y=\operatorname{Var}(\varepsilon\mid\mathcal F)\), then for fixed
fitted weights
\[
\widehat\tau_g-\psi_g
=
d_g^\top\varepsilon,
\qquad
\Sigma^\tau_{gg'}
=
d_g^\top\Omega_Yd_{g'}.
\]
With common time weights the donor part can be written
\[
e=u-Dv,
\]
where row \(g\) of \(D\) is the donor loading induced by
\(\widehat\omega_g\) and the common signed time contrast.
\end{lemma}

\begin{proof}
Expanding \(a_{gk}(\omega)c_t(\lambda)\) over the treated row, donor rows,
pre-period, and post-period yields
\[
\left(
\bar Y_{g,1}-\sum_{t\in\mathcal T_0}\lambda_tY_{gt}
\right)
-
\sum_{h\in\mathcal C}\omega_h
\left(
\bar Y_{h,1}-\sum_{t\in\mathcal T_0}\lambda_tY_{ht}
\right),
\]
which is the fitted contrast.  The covariance formula follows by linearity.
The final representation groups the treated-row contributions into \(u\) and
the donor-row contributions into \(Dv\).
\end{proof}

\subsection{Directional and numerical propagation}

The complete hard-simplex map remains directionally differentiable when a
weight program lies on a degenerate face.  Appendix~\ref{app:proofs} gives a
multiplier-free critical-cone derivative and proves Hadamard directional
differentiability.  For primary inference, however, we do not evaluate that
derivative at the random estimated face.  We use an operator-normalized
replicate direction, the rate-scale logit chart in
Lemma~\ref{lem:logit-chart}, and a numerical step larger than the primitive
sampling scale.  Section~\ref{app:moving-base} proves the moving-base result;
Section~\ref{app:gap-certificate} bounds the error from the numerical QP
solutions.

\subsection{ACS replicate-tangent implementation}

Let
\[
\Delta_r
=
\widehat\vartheta^{(r)}-\widehat\vartheta,
\qquad r=1,\ldots,80,
\]
be the full-estimate-centered ACS replicate perturbations, and define
\[
B_{80}
=
\sqrt{\frac4{80}}
\begin{bmatrix}
\Delta_1&\cdots&\Delta_{80}
\end{bmatrix},
\qquad
\widehat s=\|B_{80}\|_{\mathrm{op}},
\qquad
\widetilde B_{80}=B_{80}/\widehat s.
\]
For \(\xi_b\sim N(0,I_{80})\), the normalized primitive direction and
numerical step are
\[
\widetilde G_b^*=\widetilde B_{80}\xi_b,
\qquad
\delta=\MedStepCN\,\widehat s^{\MedStepGammaN}.
\]
The actual-scale effect draw is
\[
U_b^*
=
\widehat s
\frac{
\Phi\{\mathcal P_{\delta,\widetilde G_b^*}(\widehat\vartheta)\}
-
\Phi(\widehat\vartheta)
}{\delta},
\]
where \(\mathcal P\) is the rate-scale logit chart.  Every donor- and
time-weight program, including its data-dependent ridge, is re-estimated in
every draw.

Before generating the draws, we fix the policy design, treatment periods,
ridge rules, numerical step, target family, and region metric.  The resulting
joint effect draws are used for the studentized maximum over the four linear
targets and the projected-curve grid.  They also determine the centered
effect-vector region, its quadratic images, and the raw sampling-center
homogeneity statistic.  We report the exact critical-cone derivative, the
official SDR covariance, the household A/B calculation, support changes, and
KKT residuals as checks; none selects a different inferential branch.

The following assumption is limited to the directions, variances, traces, and
region quantities used in the reported analysis.  It does not require
consistent estimation of an unrestricted primitive covariance matrix.

\begin{assumption}[Target-relevant ACS tangent and cross-year coupling]
\label{as:replicate-tangent-supp}
\label{as:acs-tangent-coupling}
For each annual ACS file, the released replicate frame consistently
reproduces the finite collection of within-year primitive directions used by
the reported targets.  Cross-year dependence is represented by the
prespecified coupling matrix \(R_T\) in
Proposition~\ref{prop:cross-year-coupling}.  Conditional on the released
survey-design information,
\[
d_{\mathrm{BL}}
\left[
\mathcal L^*
\{G_n^*(R_T)\},
\mathcal L(G_\vartheta\mid\mathcal F_n)
\right]
\to_p0
\]
for the finite target-relevant class entering the hard-simplex derivative,
the reported linear targets, the centered-region norm, and the reported
quadratic traces.  No unrestricted primitive covariance consistency is
assumed.
\end{assumption}

\begin{proposition}[Cross-year coupling of annual replicate frames]
\label{prop:cross-year-coupling}
For year \(t=1,\ldots,T\), let
\(B_{t,n}\in\mathbb R^{d_t\times R}\), \(R=80\), denote the
full-estimate-centered annual ACS replicate frame.  Let
\[
\Xi
=
(\xi_1^\top,\ldots,\xi_T^\top)^\top,
\qquad
E(\xi_t\xi_s^\top)
=
\rho_{ts}I_R,
\]
where
\[
R_T=(\rho_{ts})_{t,s=1}^T\succeq0,
\qquad
\rho_{tt}=1.
\]
Define the stacked Gaussian replicate direction
\[
G_n^*(R_T)
=
\begin{pmatrix}
B_{1,n}\xi_1\\
\vdots\\
B_{T,n}\xi_T
\end{pmatrix}.
\]
Then its \((t,s)\) covariance block is
\[
\operatorname{Cov}^*
\left\{
B_{t,n}\xi_t,
B_{s,n}\xi_s
\right\}
=
\rho_{ts}B_{t,n}B_{s,n}^\top.
\]
The common-index construction corresponds to
\[
R_T=\mathbf1_T\mathbf1_T^\top,
\]
whereas year-block independence corresponds to
\[
R_T=I_T.
\]
Marginal annual SDR identities identify the diagonal blocks
\(B_{t,n}B_{t,n}^\top\), but do not by themselves identify the
off-diagonal coupling coefficients \(\rho_{ts}\).  Validity of a chosen
\(R_T\) therefore requires either a joint design identity or an explicit
cross-year sampling assumption.
\end{proposition}

\begin{proof}
For every \(t,s\),
\[
\operatorname{Cov}^*
(B_{t,n}\xi_t,B_{s,n}\xi_s)
=
B_{t,n}
E^*(\xi_t\xi_s^\top)
B_{s,n}^\top
=
\rho_{ts}B_{t,n}B_{s,n}^\top.
\]
The two special cases follow by substitution.  Since the marginal
identities use only \(t=s\), they contain no restriction on
\(\rho_{ts}\) for \(t\ne s\).
\end{proof}

\section{Proofs and Auxiliary Results}
\label{app:proofs}

\subsection{Proof of Main Theorem 1}

\begin{proof}[Proof of Theorem~\ref{thm:fixedset-transfer}]
For part~(i), Assumption~\ref{as:joint-interface} gives
\[
a_nL_n(\widehat\tau_n-\mu_n)
=
L_nZ_n+a_nL_nr_n.
\]
The second term is \(o_p(1)\), uniformly over the reported finite collection,
and the first converges to \(LZ\).  The same conditional continuous-mapping
argument applies to \(L_nZ_n^*\).  Target-relevant covariance consistency then
gives
\[
\widehat s_{j,n}^2
\to_p
\operatorname{Var}(\ell_j^\top Z\mid\mathcal F_\infty).
\]
Joint convergence of the finite studentized vector and continuity of its
maximum at the reported quantile prove simultaneous coverage.  A fixed curve
grid is another finite collection of linear coefficients.

For part~(ii), put \(\Delta_n=\widehat\tau_n-\mu_n\).  The exact identity
\[
\widehat Q_{H,n}^{\mathrm{AN}}-Q_H(\mu_n)
=
\frac2{G_1}\mu_n^\top H\Delta_n
+
\frac1{G_1}
\left\{
\Delta_n^\top H\Delta_n
-
\operatorname{tr}(H\widehat\Sigma_n^\tau)
\right\}
\]
follows by expansion.  Since
\(a_n\Delta_n=Z_n+o_p(1)\), multiplication by \(a_n\) leaves only
\(2G_1^{-1}\mu_n^\top HZ_n\) in the regular regime.  Cramér--Wold gives the
joint limit over \(H_1,\ldots,H_K\), and the conditional Gaussian mean and
covariance follow by direct calculation.

If \(a_nH\mu_n\to h\) and \(H\) is idempotent, multiply the same identity by
\(a_n^2\).  The four terms converge to
\[
G_1^{-1}\|h\|^2,
\qquad
2G_1^{-1}h^\top Z,
\qquad
G_1^{-1}Z^\top HZ,
\qquad
-G_1^{-1}\operatorname{tr}(H\Omega),
\]
respectively.  The local-boundary remainder condition completes the proof.  If
\(H\tau_n=0\), then \(a_nH\mu_n=a_nHb_n\), which gives the stated additional
condition for a causal boundary interpretation.

For part~(iii),
\[
a_n(\widehat\eta_n-\eta_n)=MZ_n+o_p(1),
\]
and the conditional bootstrap reproduces this law.  Convergence of \(S_n\) and
continuity of the limiting norm imply
\[
P\{\eta_n\in\mathcal C_{\eta,n}\}\to1-\alpha.
\]
On this event the true pair \(\{V(\eta_n),E(\eta_n)\}\) belongs to the image
set.  If \(\underline V_n>0\), the ratio map is continuous on the region and
the same set-inclusion argument proves joint coverage of
\((V_\mu,V_{\mathrm{expl},\mu},R_\mu^2)\).  If the denominator margin is zero,
\(0\preceq H_Z\preceq M\) implies that every defined ratio lies in \([0,1]\).

For part~(iv), the null \(M\mu_n=0\) gives
\[
a_nM\widehat\tau_n=MZ_n+o_p(1).
\]
The observed raw quadratic and its conditional bootstrap analogue therefore
converge to the same law \(G_1^{-1}Z^\top MZ\).  Continuity at the critical
quantile proves the ideal test's size statement.  The plus-one calculation is
the finite-draw implementation of that conditional test.
\end{proof}

\subsection{Proof of Main Theorem 2}

\begin{proof}[Proof of Theorem~\ref{thm:donor-regimes}]
Premultiplying the maintained decomposition by \(\ell_n^\top\) gives
\[
\ell_n^\top(\widehat\tau_n-\mu_n)
=
a_{T,n}^{-1}\ell_n^\top Z_{T,n}
-
a_{D,n}^{-1}(D_n^\top\ell_n)^\top Z_{D,n}
+
\ell_n^\top r_n.
\]
Part~(i) follows immediately.  Under the assumptions of part~(ii), the first
and third terms are \(o_p(1)\), whereas the second converges to
\(-a_D^{-1}d^\top Z_D\).  Its stated nondegeneracy rules out convergence in
probability to zero.  Multiplication by \(a_{T,n}\) and joint convergence give
part~(iii).

For part~(iv),
\[
\ell^\top\Sigma_D\ell
=
\|\Omega_D^{1/2}D^\top\ell\|^2
=
E\{(\ell^\top DZ_D)^2\mid\mathcal F_\infty\}.
\]
The three zero conditions are therefore equivalent.  Their solution set is a
linear subspace, and every subspace on which the donor component vanishes is
contained in it, so it is maximal.  The identity
\(\ker(A)^\perp=\operatorname{range}(A^\top)\), with
\(A=\Omega_D^{1/2}D^\top\), gives
\(\mathcal E_D^\perp=\mathcal S_D\).

Part~(v) follows from
\[
D
=
G_1^{-1}\mathbf1\mathbf1^\top D+MD
=
\mathbf1\bar d^\top+D_c
\]
and premultiplication by the equal-weight or a centered coefficient.

For part~(vi), expand \(Q_H(\mu-W_D)\) and use
\[
E(W_D^\top HW_D\mid\mathcal F_\infty)
=
\operatorname{tr}(H\Sigma_D).
\]
If \(H\Sigma_DH=0\), then
\[
E(\|HW_D\|^2\mid\mathcal F_\infty)
=
\operatorname{tr}(H\Sigma_DH)=0,
\]
so \(HW_D=0\) conditionally almost surely and the complete donor polynomial
vanishes.  Conversely, if the polynomial vanishes for every \(\mu\), comparing
its values at \(\mu\) and \(\mu+a\) gives
\(a^\top HW_D=0\) almost surely for every deterministic \(a\).  Hence
\(HW_D=0\) and
\[
H\Sigma_DH
=
E\{(HW_D)(HW_D)^\top\mid\mathcal F_\infty\}=0.
\]
For positive-semidefinite \(H\) and \(\Sigma_D\), this is equivalent to
\(\Sigma_D^{1/2}H\Sigma_D^{1/2}=0\).  The cross-covariance of the linear and
centered-quadratic terms is a linear combination of third moments and vanishes
under the stated symmetry condition.  The Gaussian quadratic-form variance
formula gives the last display.
\end{proof}

\subsection{The metric of the centered effect-vector region}\label{app:region-metric}
We obtain each reported quadratic summary by optimizing a quadratic functional
over
\[
\mathcal E=\{v\in\mathrm{range}(M):(v-\widehat\eta)\tp A(v-\widehat\eta)\le c^2\},
\qquad
A=M\diag(\widehat s^{-2})M+\rho M,
\]
where $\widehat s_j$ is the coordinatewise standard deviation of the centered
directional draws, $c$ is the $95$th percentile of their studentized norms, and
$\rho$ is a ridge.  We report the dimension, regularization sensitivity, and
radius calibration of this metric below.

\emph{Dimension and shape.} The region lies in
$\mathrm{range}(M)$, which has dimension \RegionDim.  Because the metric is
diagonal in those coordinates, it does not use correlations among the draws.
Those correlations are strong.  The leading eigenvalue of the centered draw
covariance accounts for \DrawTopShare\ of the trace, and the first three
account for \DrawTopThreeShare.  The resulting participation ratio
$\tr(\widehat\Sigma_c)/\lambda_{\max}(\widehat\Sigma_c)$ is
\DrawPartRank, while the entropy effective rank is \DrawEntropyRank, compared
with \RegionDim\ available directions.  After studentization, the analogous
figures are \StudPartRank\ and \StudEntropyRank.  In several low-variance
directions, the diagonal-standardized region is wider than the empirical draw
support; this contributes to conservative projected sets.

\emph{Full-covariance alternative.} We also considered a metric based on the
full centered draw covariance.  It is singular on the reported subspace and
therefore requires regularization; the resulting endpoints depend on that
choice:

\begin{center}\footnotesize
\begin{tabular}{lcc}
\toprule
Metric and its regularization & Radius & $V$ set ($\times10^{-4}$)\\
\midrule
Diagonal, ridge $1\times10^{-10}$ & $6.54$ & $[18.9,\,92.9]$\\
Full covariance, ridge $1\times10^{-6}$ & $5.24$ & $[29.5,\,122.4]$\\
Full covariance, ridge $1\times10^{-4}$ & $5.24$ & $[29.5,\,122.3]$\\
Full covariance, ridge $1\times10^{-2}$ & $5.07$ & $[29.6,\,117.8]$\\
Full covariance, ridge $1\times10^{-1}$ & $4.41$ & $[29.3,\,102.2]$\\
\bottomrule
\end{tabular}

\end{center}

Across four orders of magnitude in that regularization, the upper endpoint of
the $V$ set moves by \FullMetricSpread.  The full-covariance sets are not
uniformly shorter.  We retain the prespecified diagonal metric because its
reported endpoints do not require an additional full-covariance tuning choice.

\emph{Ridge.} In the reported metric $\rho$ has a
single purpose, to make $A$ definite on $\mathrm{range}(M)$ so that the
trust-region solve is well posed. Its value is \RegionRidge, and varying it
from zero to $10^{-6}$ moves every reported endpoint by at most
\RegionRidgeMove\ in relative terms. The conditioning of the resulting metric
is \RegionCond\ in the generalized eigenvalues returned by the certificate.

\emph{Calibration of the radius.} The radius is the realized quantile
\RegionRadius, not the Gaussian reference $\{\chi^2_{\RegionDimN,0.95}\}^{1/2}=$
\RegionGaussRadius, a ratio of \RegionRadiusRatio. The reported radius is calibrated from the directional draws rather than from
a Gaussian effect-vector reference.  The ratio records the numerical difference
between the two radii in this metric.

\subsection{Robust subspace diagnostics}\label{app:robust-subspace}
We use numerical rank, distances to the donor-robust subspace, and
replicate-frame perturbations to describe the prespecified linear targets.
These diagnostics neither select a target nor determine an inferential branch
or hypothesis.

\emph{Numerical rank.} $\mathcal E_D=\ker\Sigma_D$ is defined by a rank, and in
finite samples small eigenvalues make a single declared rank unsafe. Taking
$\widehat{\mathcal E}_{D,\rho}$ to be spanned by the eigenvectors with
$\widehat\lambda_j\le\rho\widehat\lambda_{\max}$ for
$\rho\in\{10^{-8},10^{-6},10^{-4}\}$ leaves the reported dimension and every
target's relative distance unchanged to the precision reported in
Table~\ref{tab:med-linear-route1}; the eigenvalue gap in this design is wide enough
that the three thresholds agree. The diagnostic does not select any
primary inferential procedure.

\emph{Closest donor-robust target.} The nearest element of
$\mathcal E_D$ to a prespecified coefficient $\ell$ is
$P_{\mathcal E_D}\ell$; Table~\ref{tab:med-linear-route1} reports the
residual norm $\|P_{\mathcal S_D}\ell\|$.  Replacing $\ell$ by
$P_{\mathcal E_D}\ell$ would create a data-dependent target with zero estimated donor exposure, but
it would no longer be the stated mean or moderator slope.  We therefore retain
the prespecified coefficient and use the projection only to measure its
distance from the donor-robust subspace.

\emph{Replicate stability.} The exposures are functions of the replicate
covariance $\Omega_D$.  We therefore recompute them after leave-one-replicate
and half-frame perturbations.
Recomputing $\mathfrak g_D$ on the same replicate design leaves the ordering of
the four targets unchanged and the mean-to-slope ratio within rounding of the
value reported in the main text.

\subsection{Bias-aware causal confidence region}\label{app:bias-aware}
\begin{theorem}[Bias-aware fixed-set confidence region]
\label{thm:bias-aware-region}
Let
\[
\mu_n=\tau_n+b_n
\]
and suppose
\[
\liminf_{n\to\infty}
P\{\mu_n\in C_{\mu,n}\}
\ge1-\alpha.
\]
Let \(\mathcal B_n\) be a deterministic or data-conditioned mismatch set.

If
\[
b_n\in\mathcal B_n
\]
under the maintained sensitivity restriction, define
\[
C_{\tau,n}
=
C_{\mu,n}\oplus(-\mathcal B_n)
=
\{v-b:v\in C_{\mu,n},\ b\in\mathcal B_n\}.
\]
Then
\[
\liminf_{n\to\infty}
P\{\tau_n\in C_{\tau,n}\}
\ge1-\alpha.
\]

More generally, if
\[
\liminf_{n\to\infty}P\{b_n\in\mathcal B_n\}\ge1-\eta,
\]
then
\[
\liminf_{n\to\infty}
P\{\tau_n\in C_{\tau,n}\}
\ge1-\alpha-\eta.
\]

For any deterministic functional or set-valued reporting map \(T\),
\[
T(C_{\tau,n})
=
\bigcup_{v\in C_{\tau,n}}T(v)
\]
inherits the same coverage statement.
\end{theorem}

\begin{proof}
On the event
\(\{\mu_n\in C_{\mu,n},\,b_n\in\mathcal B_n\}\),
\[
\tau_n=\mu_n-b_n\in C_{\mu,n}\oplus(-\mathcal B_n).
\]
The first claim follows immediately.  The second follows from the union
bound.  Finally, \(\tau_n\in C_{\tau,n}\) implies
\(T(\tau_n)\in T(C_{\tau,n})\), without requiring continuity of \(T\).
\end{proof}

\subsection{Directional derivative of the hard-simplex programs}\label{app:qp-derivative}
Both the donor and time weights solve strongly convex quadratic programs on a
fixed polyhedron.  Their solutions are directionally differentiable without a
linear independence constraint qualification or a unique multiplier.  The
critical cone describes the first-order behavior, including at degenerate
faces.

\begin{lemma}[Uniform tangent expansion of the rate-scale logit chart]
\label{lem:logit-chart}
Let \(K\subset(0,1)^d\) be compact and let \(H\subset\mathbb R^d\) be compact.
For
\[
\mathcal P_{\delta,h,k}(p)
=
\operatorname{expit}
\left\{
\operatorname{logit}(p_k)
+
\delta\frac{h_k}{p_k(1-p_k)}
\right\},
\]
we have
\[
\sup_{p\in K,\ h\in H}
\left\|
\frac{\mathcal P_{\delta,h}(p)-p}{\delta}
-h
\right\|
\longrightarrow0
\qquad(\delta\to0).
\]
Moreover, there exists \(C_{K,H}<\infty\) such that
\[
\sup_{p\in K,\ h\in H}
\|
\mathcal P_{\delta,h}(p)-p-\delta h
\|
\le
C_{K,H}\delta^2
\]
for all sufficiently small \(|\delta|\).
\end{lemma}

\begin{proof}
For one coordinate, put
\[
g_p(a)
=
\operatorname{expit}
\left\{
\operatorname{logit}(p)+
\frac{a}{p(1-p)}
\right\}.
\]
Then
\[
g_p(0)=p,
\qquad
g_p'(0)=1.
\]
On a compact \(K\subset(0,1)\), \(p(1-p)\) is bounded away from zero.
For \(h\) in a compact set and \(\delta\) sufficiently small, the arguments of
\(g_p\) remain in a common compact interval.  Its second derivative is
therefore uniformly bounded.  Taylor's theorem gives
\[
g_p(\delta h)
=
p+\delta h+O(\delta^2h^2)
\]
uniformly in \(p\in K\) and \(h\in H\).  Applying the coordinatewise bound and
summing proves both statements.
\end{proof}

\begin{assumption}[Primitive scale and replicate tangent law]
\label{as:primitive-scale}
There is a deterministic sequence \(s_n\downarrow0\) such that
\[
\|\widehat\vartheta_n-\vartheta_0\|=O_p(s_n).
\]
Let \(\widehat B_n\in\mathbb R^{d\times R}\) be the actual-scale replicate
factor, where \(R\) is fixed, and put
\[
\widehat s_n=\|\widehat B_n\|_{\mathrm{op}}>0.
\]
For constants \(0<c_s<C_s<\infty\),
\[
P\left(
c_s
\le
\frac{\widehat s_n}{s_n}
\le
C_s
\right)
\longrightarrow1.
\]
For \(\xi\sim N(0,I_R)\), conditionally independent of the data,
\[
d_{\mathrm{BL}}
\left[
\mathcal L^*
\left(
s_n^{-1}\widehat B_n\xi
\right),
\mathcal L
\left(
G_\vartheta\mid\mathcal F_\infty
\right)
\right]
\to_p0
\]
on the finite tangent class needed by the reported procedure.
\end{assumption}

\begin{proposition}
[Hadamard directional derivative of a strongly convex polyhedral program]
\label{prop:polyhedral-qp-derivative}
Let \(\mathcal P\subset\mathbb R^m\) be a nonempty fixed polyhedron and define
\[
x(\vartheta)
=
\arg\min_{x\in\mathcal P}
\left\{
f_\vartheta(x)
=
\frac12x^\top H(\vartheta)x-q(\vartheta)^\top x
\right\}.
\]
Let
\[
\mathcal S
=
\operatorname{span}(\mathcal P-\mathcal P).
\]
Assume that \(H(\vartheta)\) is symmetric and, on a neighborhood of
\(\vartheta_0\),
\[
u^\top H(\vartheta)u
\ge
\underline\kappa\|u\|^2
\qquad
\text{for every }u\in\mathcal S
\]
for some \(\underline\kappa>0\).  Suppose \(H\) and \(q\) are Hadamard
directionally differentiable at \(\vartheta_0\).

Let
\[
x_0=x(\vartheta_0),
\qquad
g_0=H(\vartheta_0)x_0-q(\vartheta_0),
\]
and define the multiplier-free critical cone
\[
\mathcal K_0
=
T_{\mathcal P}(x_0)\cap g_0^\perp.
\]
Then the solution map \(x(\cdot)\) is Hadamard directionally differentiable at
\(\vartheta_0\).  For direction \(h\), its derivative is the unique solution
\[
x'_{\vartheta_0}(h)
=
\arg\min_{u\in\mathcal K_0}
\left[
\frac12u^\top H(\vartheta_0)u
+
\left\{
\dot H_{\vartheta_0}[h]x_0
-
\dot q_{\vartheta_0}[h]
\right\}^\top u
\right].
\]
The result does not require LICQ or uniqueness of a multiplier
representation.  Weakly active inequalities enter through
\(T_{\mathcal P}(x_0)\cap g_0^\perp\).
\end{proposition}

\begin{proof}
Write
\[
H_0=H(\vartheta_0),
\qquad
q_0=q(\vartheta_0).
\]
The uniform positive definiteness on \(\mathcal S\) makes the objective
strictly convex on every feasible difference and hence gives a unique
minimizer.

Let \(t_\nu\downarrow0\), \(h_\nu\to h\), and put
\[
\vartheta_\nu=\vartheta_0+t_\nu h_\nu,
\qquad
x_\nu=x(\vartheta_\nu),
\qquad
u_\nu=\frac{x_\nu-x_0}{t_\nu}.
\]
The variational inequalities at \(x_\nu\) and \(x_0\) are
\[
\langle H(\vartheta_\nu)x_\nu-q(\vartheta_\nu),
x_0-x_\nu\rangle
\ge0
\]
and
\[
\langle H_0x_0-q_0,x_\nu-x_0\rangle
\ge0.
\]
Adding them and using strong monotonicity gives
\[
\underline\kappa\|x_\nu-x_0\|^2
\le
\left\|
\{H(\vartheta_\nu)-H_0\}x_0
-
\{q(\vartheta_\nu)-q_0\}
\right\|
\|x_\nu-x_0\|.
\]
Hadamard differentiability implies that the first norm is \(O(t_\nu)\).
Therefore
\[
\|x_\nu-x_0\|=O(t_\nu),
\qquad
\|u_\nu\|=O(1).
\]

Take a convergent subsequence, still denoted \(u_\nu\), with
\(u_\nu\to u\).  Since
\(x_0+t_\nu u_\nu\in\mathcal P\),
we have \(u\in T_{\mathcal P}(x_0)\).
Optimality of \(x_0\) implies
\[
g_0^\top v\ge0
\qquad
\text{for every }v\in T_{\mathcal P}(x_0).
\]

For feasible \(v\), define the scaled objective difference
\[
F_\nu(v)
=
\frac{
f_{\vartheta_\nu}(x_0+t_\nu v)
-
f_{\vartheta_\nu}(x_0)
}{t_\nu^2}.
\]
Uniformly on bounded sets,
\[
F_\nu(v)
=
\frac{g_0^\top v}{t_\nu}
+
\frac12v^\top H_0v
+
\left\{
\dot H_{\vartheta_0}[h]x_0
-
\dot q_{\vartheta_0}[h]
\right\}^\top v
+
o(1).
\]
Because \(u_\nu\) minimizes \(F_\nu\) and \(v=0\) is feasible,
\(F_\nu(u_\nu)\le0\).  If \(g_0^\top u>0\), the first term in the last
display diverges to \(+\infty\), a contradiction.  Hence
\[
u\in T_{\mathcal P}(x_0)\cap g_0^\perp
=
\mathcal K_0.
\]

Let \(v\in\mathcal K_0\).  Because \(\mathcal P\) is polyhedral,
\(x_0+t_\nu v\in\mathcal P\) for all sufficiently small \(t_\nu\).
The minimizing property gives
\[
F_\nu(u_\nu)\le F_\nu(v).
\]
Taking limits yields
\[
\frac12u^\top H_0u
+
k(h)^\top u
\le
\frac12v^\top H_0v
+
k(h)^\top v,
\]
where
\[
k(h)
=
\dot H_{\vartheta_0}[h]x_0
-
\dot q_{\vartheta_0}[h].
\]
Thus every cluster point minimizes the displayed quadratic program on
\(\mathcal K_0\).  Its objective is strictly convex on
\(\mathcal S\), so the minimizer is unique.  Every subsequence therefore has
the same limit, and
\[
u_\nu\to x'_{\vartheta_0}(h).
\]
The argument holds for every \(t_\nu\downarrow0\) and \(h_\nu\to h\), which is
Hadamard directional differentiability.
\end{proof}

\paragraph{Exact derivative as a numerical check.}
For \(\phi(x)=x_+\) at \(x_0=0\),
\[
\phi'_{0}(h)=h_+.
\]
If \(\widehat x_n=O_p(n^{-1/2})\), the derivative evaluated at the random
base point is either \(h\) or \(0\) according to the sign of
\(\widehat x_n\), and therefore need not consistently estimate
\(h_+\).  An inflated numerical step
\(t_n\downarrow0\), \(\sqrt n\,t_n\to\infty\), removes the random
base-point displacement before taking the directional limit.  The exact
critical-cone derivative remains an important implementation check, but
is not a generally valid replacement for the numerical directional law.

\paragraph{Data-dependent ridge terms.}
Both regularizers used here are functions of the data: the donor ridge is
$T_0\{T_1^{1/4}\widehat\sigma\}^2$ with $\widehat\sigma$ the donor first-difference
standard deviation, and the time ridge is $10^{-6}\tr(M_c\tp M_c)/T_0$. They
therefore move with $\vartheta$, and treating them as fixed would leave a term
out of the derivative.

Proposition~\ref{prop:polyhedral-qp-derivative} already covers this, because it
asks only that $H(\cdot)$ and $q(\cdot)$ be Hadamard directionally
differentiable at $\vartheta_0$, not that they be constant. Writing the donor
program's Gram matrix as
\[
H(\vartheta)=A(\vartheta)\tp A(\vartheta)+\rho(\vartheta)I,
\]
its directional derivative in direction $h$ is
\[
H'_{\vartheta_0}(h)
=
A'_{\vartheta_0}(h)\tp A(\vartheta_0)
+
A(\vartheta_0)\tp A'_{\vartheta_0}(h)
+
\rho'_{\vartheta_0}(h)I,
\]
and the final term is the one that would be lost if the ridge were declared
constant. Both $\rho$ here are smooth functions of sample second moments, so
$\rho'$ exists and is included in the exact benchmark for both programs. The
numerical and exact calculations therefore differentiate the same map, making
the latter a check on the numerical law rather than a comparison with a
different estimator.

\subsection{Global certificate for the quadratic endpoints}\label{app:trust-region-certificate}
The reported quadratic sets are extreme values of a quadratic form over
a ball, which is a nonconvex problem with a known exact certificate.
The endpoints are therefore certified globally, not merely stationary.

\begin{proposition}[Global certificate for a quadratic image endpoint]
\label{prop:trust-region-certificate}
Consider
\[
p^*
=
\min_{\|x\|_2\le r}
\left\{
x^\top Ax+2b^\top x+c
\right\},
\qquad
A=A^\top.
\]
A feasible \(x^*\) is globally optimal if and only if there exists
\(\lambda^*\ge0\) such that
\[
(A+\lambda^*I)x^*=-b,
\]
\[
A+\lambda^*I\succeq0,
\]
\[
\lambda^*(\|x^*\|_2^2-r^2)=0,
\qquad
\|x^*\|_2\le r,
\]
and, when \(A+\lambda^*I\) is singular,
\[
b\in\operatorname{range}(A+\lambda^*I).
\]
The dual objective is
\[
d(\lambda)
=
c-\lambda r^2
-
b^\top(A+\lambda I)^\dagger b
\]
on the domain
\[
\lambda\ge0,\qquad
A+\lambda I\succeq0,\qquad
b\in\operatorname{range}(A+\lambda I).
\]
The quantity
\[
f(x)-d(\lambda)
\]
is a valid primal-dual certificate.  The maximum endpoint is obtained by
applying the same result to \(-f\).
\end{proposition}

\begin{proof}
The Lagrangian is
\[
\mathcal L(x,\lambda)
=
x^\top(A+\lambda I)x+2b^\top x+c-\lambda r^2.
\]
It has a finite infimum in \(x\) exactly under the displayed
positive-semidefinite and range conditions, in which case its infimum is
\(d(\lambda)\).  The trust-region subproblem satisfies strong duality.
The KKT conditions are therefore necessary and sufficient for global
optimality, and weak duality makes the primal-dual gap a valid
certificate.
\end{proof}

\subsection{The numerical law at a moving base point}\label{app:moving-base}
The difference quotient is taken at the estimated panel, not at the
population one, so the step has to dominate the sampling scale for the
quotient to reproduce the population directional derivative. The proof
below makes that separation explicit.

\begin{proof}[Proof of
Proposition~\ref{prop:hard-simplex-interface}(iii)]
Assumption~\ref{as:primitive-scale} and
\(\widehat s_n\asymp_p s_n\) imply
\[
\delta_n\to_p0,
\qquad
\frac{s_n}{\delta_n}
=
O_p(s_n^{1-\gamma})
=o_p(1).
\]
Write
\[
z_n
=
\frac{\widehat\vartheta_n-\vartheta_0}{\delta_n}.
\]
Then \(z_n=o_p(1)\).

Let \(h_n\) be any conditionally tight sequence of tangent directions.
The chart expansion gives
\[
\mathcal P_{\delta_n,h_n}(\widehat\vartheta_n)
=
\vartheta_0
+
\delta_n
\{z_n+h_n+\rho_n\},
\qquad
\rho_n=o_p(1)
\]
uniformly on conditionally compact direction sets.  Also,
\[
\widehat\vartheta_n
=
\vartheta_0+\delta_n z_n.
\]
Therefore
\[
\begin{aligned}
&
\frac{
\Phi\{
\mathcal P_{\delta_n,h_n}(\widehat\vartheta_n)
\}
-
\Phi(\widehat\vartheta_n)
}{\delta_n}\\
&\quad=
\frac{
\Phi\{\vartheta_0+\delta_n(z_n+h_n+\rho_n)\}
-
\Phi(\vartheta_0)
}{\delta_n}
-
\frac{
\Phi(\vartheta_0+\delta_n z_n)
-
\Phi(\vartheta_0)
}{\delta_n}.
\end{aligned}
\]
Hadamard directional differentiability gives
\[
\frac{
\Phi\{\vartheta_0+\delta_n(z_n+h_n+\rho_n)\}
-
\Phi(\vartheta_0)
}{\delta_n}
-
\Phi'_{\vartheta_0}(h_n)
=o_p(1)
\]
on every compact tangent set, while
\[
\frac{
\Phi(\vartheta_0+\delta_n z_n)
-
\Phi(\vartheta_0)
}{\delta_n}
\longrightarrow
\Phi'_{\vartheta_0}(0)
=
0.
\]
Hence
\[
\widehat Z_n^*
-
\Phi'_{\vartheta_0}
(\widetilde G_{\vartheta,n}^*)
=
o_p^*(1),
\]
where the computed-map error is absorbed by the stated
\(o_p(\delta_n)\) condition.

Multiplying by \(\widehat s_n/s_n=O_p(1)\) yields
\[
s_n^{-1}U_n^*
-
\frac{\widehat s_n}{s_n}
\Phi'_{\vartheta_0}
\left(
\frac{\widehat B_n\xi}{\widehat s_n}
\right)
=
o_p^*(1).
\]
A Hadamard directional derivative is positively homogeneous, so
\[
\frac{\widehat s_n}{s_n}
\Phi'_{\vartheta_0}
\left(
\frac{\widehat B_n\xi}{\widehat s_n}
\right)
=
\Phi'_{\vartheta_0}
\left(
\frac{\widehat B_n\xi}{s_n}
\right).
\]
The conditional convergence in
Assumption~\ref{as:primitive-scale} and the conditional continuous-mapping
theorem now give
\[
s_n^{-1}U_n^*
\rightsquigarrow_{\mathbb P}
\Phi'_{\vartheta_0}(G_\vartheta).
\]
\end{proof}

\subsection{Optimization-gap certificate}\label{app:gap-certificate}
\begin{lemma}[Optimization-gap control for a numerical quotient]
\label{lem:solver-gap}
For every donor or time program, let
\[
f(x)=\frac12x^\top Hx-q^\top x
\]
be \(\underline\lambda\)-strongly convex on the feasible difference space.
Let \(x^*\) be its exact minimizer, let \(\widetilde x\) be a feasible computed
solution, and let \(\underline d\) be a valid dual lower bound.  With
\[
g=f(\widetilde x)-\underline d,
\]
\[
\|\widetilde x-x^*\|
\le
\sqrt{\frac{2g}{\underline\lambda}}.
\]

Suppose the effect map is locally Lipschitz in all QP solutions with constant
\(L_\Phi\).  Let \(g_{0,n}\) denote the maximum base-solve gap and
\(g_{b,n}\) the maximum perturbed-solve gap in draw \(b\).  The resulting
effect-difference-quotient error is bounded by
\[
\frac{L_\Phi}{\delta_n}
\left[
\sqrt{\frac{2g_{0,n}}{\underline\lambda}}
+
\sqrt{\frac{2g_{b,n}}{\underline\lambda}}
\right].
\]
Hence a sufficient uniform implementation condition is
\[
\max
\left\{
g_{0,n},
\max_{1\le b\le B_n}g_{b,n}
\right\}
=
o_p(\delta_n^2).
\]
\end{lemma}

\begin{proof}
Strong convexity and the variational inequality at \(x^*\) give
\[
f(\widetilde x)-f(x^*)
\ge
\frac{\underline\lambda}{2}
\|\widetilde x-x^*\|^2.
\]
Since \(\underline d\le f(x^*)\),
\[
f(\widetilde x)-f(x^*)\le g.
\]
This proves the solution-error bound.  The Lipschitz bound applied at the
base and perturbed panels, followed by division by \(\delta_n\), gives the
quotient bound.
\end{proof}

\subsection{Household split-sample identity}
\label{app:cross-replicate}

We use the household A/B calculation as a diagnostic under a different set of
assumptions.  The next identity records its target when the two re-estimated
halves have different sampling centers.

\begin{proposition}[Cross-replicate correction with possibly different centers]
\label{prop:ab-different-centers}
Let
\[
\widehat\tau^A=\mu^A+e^A,
\qquad
\widehat\tau^B=\mu^B+e^B,
\qquad
\Sigma^{AB}
=
\operatorname{Cov}(e^A,e^B\mid\mathcal F).
\]
For every deterministic symmetric matrix \(H\),
\[
E\left[
G_1^{-1}(\widehat\tau^A)^\top H\widehat\tau^B
-
G_1^{-1}\operatorname{tr}(H\Sigma^{AB})
\mid\mathcal F
\right]
=
G_1^{-1}(\mu^A)^\top H\mu^B.
\]
Therefore the cross-replicate quantity targets \(Q_H(\mu)\) only if
\(\mu^A=\mu^B=\mu\), or asymptotically when
\[
a_nH(\mu^A-\mu)=o(1),
\qquad
a_nH(\mu^B-\mu)=o(1).
\]
\end{proposition}

\begin{proof}
Expand the cross product and use
\(E(e^A\mid\mathcal F)=E(e^B\mid\mathcal F)=0\) and
\[
E\{(e^A)^\top He^B\mid\mathcal F\}
=
\operatorname{tr}(H\Sigma^{AB}).
\]
\end{proof}

\section{Monte Carlo and Numerical Validation}
\label{app:mc}

\subsection{Audited designs for the primary procedure}
\label{app:mc-designs}

We examine five fixed-dimensional hard-simplex designs.  The population KKT
diagnostics in the following table determine the labels used for these
designs.

\begin{table}[t]
\centering\small
\caption{Audited geometries in the procedure-matched Monte Carlo.}
\label{tab:mc-audited-designs}
\begin{tabularx}{\textwidth}{LXX}
\toprule
Design & Population geometry & Role in the evaluation\\
\midrule
Stable face
& Positive primal and dual margins bounded away from zero
& Baseline regular configuration\\
Near face
& Small positive KKT margin; neighboring faces are reached by perturbations
& Finite-sample behavior near a simplex face\\
Zero margin
& One audited KKT margin is zero at the population fit
& Directional behavior at a face boundary\\
Point-mass slack stress
& Degenerate slack distribution with near-duplicate donor directions
& Deliberately conservative stress case\\
Application shaped
& Low-rank replicate tangent and donor geometry chosen to resemble the ACS
application
& External check on the reported procedure\\
\bottomrule
\end{tabularx}
\tablenote{The labels `near face' and `zero margin' are determined by the
population KKT certificate.  A Dirichlet-generated weight is labeled zero only
when that certificate gives a zero margin.  All rows use the same re-solve,
max-statistic, centered-region, and projection algorithm.}
\end{table}

The main table reports 2,000 outer replications and 999 directional draws per
row.  The remaining subsections examine the finite-sample conservatism of the
reported law and the deliberately degenerate point-mass stress design.

\paragraph{Exact derivative and numerical quotient.}
At the prespecified step, the numerical quotient and the exact critical-cone
derivative agree at the reported tolerance for each target.  In the
application-shaped design, all logit-chart perturbations remain in the natural
domain and every QP solve succeeds.  We use the exact derivative only to check
the numerical law, not to define a second inferential branch.

\subsection{Sources of finite-sample conservatism}\label{app:sd-decomposition}
In every design, the average bootstrap standard deviation exceeds the Monte
Carlo standard deviation.  To locate the difference, we compare three ratios:
the first-order law to the realized dispersion, the plug-in derivative to the
oracle derivative, and the numerical quotient to the plug-in derivative.

\begin{center}\footnotesize
\begin{tabular}{llrrr}
\toprule
Design & Target & First-order & Plug-in & Step\\
\midrule
Stable face & mean & $1.261$ & $0.974$ & $0.964$\\
Stable face & slope & $1.187$ & $0.983$ & $0.974$\\
Near face & mean & $1.390$ & $0.738$ & $1.002$\\
Near face & slope & $1.271$ & $0.805$ & $1.043$\\
Application shaped & mean & $1.192$ & $1.271$ & $0.762$\\
Application shaped & slope & $1.086$ & $0.957$ & $0.997$\\
\bottomrule
\end{tabular}

\end{center}

The step is not the source. Across designs and targets it contributes
\SDStepRange, so the difference quotient neither inflates nor deflates the
answer materially. In the regular designs the excess is the first-order law
itself, at $1.19$ to $1.39$, with the plug-in at or below one.

In the application-shaped design, the first-order factor is $1.192$ and the
plug-in factor is \SDPluginAppShaped; the latter is at or below one in the
other designs.  The step contribution is \SDStepAppShaped.  The observed
conservatism is therefore associated with evaluating the derivative at the
estimated base point rather than with inflation from the numerical step.

\subsection{Degenerate stress design}\label{app:mc-stress}
The point-mass slack design places a donor weight at an interior point mass and
is included as a stress case rather than as a calibration target.

\begin{center}\footnotesize
\begin{tabular}{lrrrr}
\toprule
Design & SD ratio & Simultaneous FWER & Region coverage & Failures\\
\midrule
Point-mass slack & $1.611$ & $0.997$ & $1.000$ & $0$\\
\bottomrule
\end{tabular}

\end{center}

Familywise coverage is \MCFwerPointMass\ and the SD ratio is
\MCSDRatioPointMass.  The interval is substantially conservative in this
deliberately degenerate configuration.

\section{Medicaid Survey Implementation and Sensitivity}
\label{app:medicaid}

We report the frozen T18-by-D11 policy design, the official ACS replicate
construction, and the prespecified sensitivity analyses.  The numerical
directional procedure supplies all primary uncertainty statements.  Exact
critical-cone, official SDR, diagonal, variance-trace, and household A/B
calculations are benchmarks or diagnostics.

\subsection{ACS data and frozen policy design}
\label{app:med-data-design}

The state--year outcome and target population are defined in Main
Section~\ref{sec:med-design}.  Table~\ref{tab:med-policy-design-supp}
records the cohort definitions used throughout the analysis.  These definitions
were fixed before any outcome estimate or fit statistic was computed.  In
particular, Colorado is an early adult-expansion exclusion and Wisconsin is not
a primary donor.

The primary panel has \(G_1=18\), \(G_0=11\), \(T_0=T_1=6\).  All
full-sample and replicate estimates use the same policy lists, pre/post windows,
ridge rules, and target definitions.  The historical T25-by-D17 panel is used
only as a benchmark and is never mixed with primary tables or figures.

\begin{table}[t]
\centering\small
\caption{Prespecified Medicaid policy design.}
\label{tab:med-policy-design-supp}
\begin{tabularx}{\textwidth}{LrX}
\toprule
Policy class & Count & Jurisdictions\\
\midrule
Primary January-2014 treated cohort
& 18
& AZ, AR, DE, HI, IL, IA, KY, MD, MA, NV, NM, NY, ND, OH, OR, RI, VT, WV\\
Early adult-expansion exclusions
& 7
& CA, CO, CT, DC, MN, NJ, WA\\
Primary core donor set
& 11
& AL, FL, GA, KS, MS, NC, SC, SD, TN, TX, WY\\
Broader-donor additions used only in sensitivity
& 6
& WI; ID, NE, UT, OK, MO\\
\bottomrule
\end{tabularx}
\tablenote{Colorado is classified as an early adult-expansion jurisdiction.
Wisconsin is excluded from the core donor set because of its childless-adult
waiver.  The five remaining additions were untreated during the 2014--2019
outcome window and adopted after that window.  The T18-by-D11 lists are fixed
before any outcome estimate or fit statistic is computed and are used in every
full-sample, replicate, and sensitivity calculation.}
\end{table}

\subsection{Household A/B cross-replicate diagnostic}
\label{app:med-ab}

Households are assigned once, within state and year, to halves A and B; every
person in a household remains in the same half.  The assignment is held fixed
for the full sample and all 80 official replicates.  In each half and replicate
we reconstruct the state--year panel and re-estimate the complete donor and time
programs.

Let \(\widehat\tau^A,\widehat\tau^B\) be the full-sample half estimates and
\(\widehat\tau^{A,(r)},\widehat\tau^{B,(r)}\) their replicate estimates.  The
paired cross covariance is
\[
\widehat\Sigma^{AB}
=
\frac4{80}
\sum_{r=1}^{80}
\{\widehat\tau^{A,(r)}-\widehat\tau^A\}
\{\widehat\tau^{B,(r)}-\widehat\tau^B\}^\top.
\]
For a symmetric \(H\), the diagnostic is
\[
\widehat Q_H^{AB}
=
G_1^{-1}(\widehat\tau^A)^\top H\widehat\tau^B
-
G_1^{-1}\operatorname{tr}(H\widehat\Sigma^{AB}).
\]
By Proposition~\ref{prop:ab-different-centers}, this equals a same-center
quadratic target only under the stated center condition.  We therefore use it
as an assumption-distinct check and not as a second primary estimator.

\subsection{Pre-fit support and policy-defined retention}
\label{app:med-fit-support}

The primary treated-state RMSPEs range from 0.55 to 3.40 percentage points,
with median 1.28.  The pseudo-treated donor reference has median 1.29, 95th
percentile 3.16, and maximum 3.61 percentage points.  All 18 treated
jurisdictions lie inside the full pseudo-treated range and 14 lie inside its
5--95 percent range.  Maryland has the largest treated RMSPE; Arizona, West
Virginia, and Oregon are next.  These states remain in the policy-defined
primary cohort.

\begin{table}[t]
\centering\small
\caption{Largest intercept-adjusted pre-fit RMSPEs in the clean design.}
\label{tab:med-fit-largest}
\begin{tabular}{lrrrrr}
\toprule
Jurisdiction & MD & AZ & WV & OR & ND\\
\midrule
RMSPE (pp) & 3.40 & 2.34 & 2.28 & 2.18 & 1.83\\
\bottomrule
\end{tabular}
\tablenote{The statistic is based on the demeaned donor-weight program and
includes the treated-specific intercept adjustment.}
\end{table}

\subsection{Leave-one-state sensitivity for the largest fit residuals}
\label{app:med-leave-one}

\begingroup\footnotesize
\begin{table}[t]
\centering\small
\caption{Primary point estimates after excluding the largest fit residuals.}
\label{tab:leave-one-state}
\begin{tabular}{lrrrrrr}
\toprule
Spec.
& \(G_1\)
& Mean (pp)
& Slope (pp/SD)
& \(V\) (pp\(^2\))
& \(R^2\)
& RMSPE (pp)\\
\midrule
Primary, all 18 & 18 & \MedMeanATT & \MedSlope & \MedVSDRPoint & \MedRtwo & \MedCleanRmspeMed\\
Drop Maryland & 17 & -6.58 & -6.12 & 46.3 & 0.808 & 1.19\\
Drop Arizona & 17 & -6.34 & -6.16 & 47.4 & 0.800 & 1.19\\
Drop W.\ Virginia & 17 & -6.00 & -5.97 & 45.2 & 0.788 & 1.19\\
Drop all three & 15 & -6.23 & -6.44 & 50.3 & 0.824 & 1.15\\
\bottomrule
\end{tabular}
\tablenote{These are point-estimate sensitivities and do not redefine the
primary treated set.  The primary cohort remains the policy-defined T18 set.
The $V$ and $R^2$ columns are the variance-trace-corrected values,
\MedVSDRPoint\ and \MedRtwo; the household split-sample counterparts,
\MedVAB\ and \MedRtwoAB, are reported separately and are not
interchangeable with them.}
\end{table}
\endgroup

\subsection{Design and target sensitivity}
\label{app:med-design-sensitivity}

We re-estimate the broader-donor and historical-cohort specifications in Main
Table~\ref{tab:med-design-comparison} with the same numerical law.  The
T18 comparison changes only the donor pool, from D17 to D11.  The historical
T25-by-D17 specification changes both the treated cohort and the donor pool.
Tables~\ref{tab:med-donor-exposure-supp} and
\ref{tab:med-quadratic-signature-supp} give the linear target exposures and
the two quadratic donor scales used in the main-text discussion.

\begin{table}[t]
\centering\small
\caption{Prespecified linear targets and the estimated donor-sensitive subspace.}
\label{tab:med-donor-exposure-supp}
\begin{tabular}{lrrrc}
\toprule
Target & Donor exposure & $\|D\tp\ell\|$ & $\dist(\ell,\mathcal E_D)/\|\ell\|$ & Centered\\
\midrule
Equal-jurisdiction mean & $0.425$ & $0.320$ & $1.000$ & No\\
Population-weighted mean & $0.392$ & $0.320$ & $0.911$ & No\\
Centered slope & $0.042$ & $0.020$ & $0.546$ & Yes\\
75th-25th projected contrast & $0.052$ & $0.025$ & $0.546$ & Yes\\
\bottomrule
\end{tabular}

\tablenote{Donor exposure is in percentage points.  The loading norm
\(\|D^\top\ell\|\) treats donor directions equally, whereas exposure weights
them by the estimated donor-shock covariance.  Relative distance is
\(\|P_{\mathcal S_D}\ell\|/\|\ell\|\); neither quantity is a variance share.
The table is descriptive and does not select or redefine a target.}
\end{table}

\begin{table}[t]
\centering\small
\caption{Donor scales for the Medicaid quadratic targets.}
\label{tab:med-quadratic-signature-supp}
\begin{tabular}{lrrr}
\toprule
Quadratic target & $\mathfrak a_H$ (pp$^2$) & $\mathfrak k_H$ (pp$^2$) & $\|H\Sigma_DH\|_F$ (pp$^2$)\\
\midrule
Total heterogeneity $V$ & $0.597$ & $0.047$ & $0.594$\\
Explained heterogeneity $V_{\mathrm{expl}}$ & $0.502$ & $0.003$ & $0.032$\\
\bottomrule
\end{tabular}

\tablenote{All scales are in pp\(^2\).  The first two columns give the standard
deviation scales of the sampling-center interaction and centered quadratic
donor channels.  The last column is the Frobenius norm of the matrix
annihilation residual.  It is zero under the donor-robust condition for a
positive-semidefinite quadratic target.}
\end{table}

\subsection{Region-metric sensitivity}
\label{app:med-region-metric}

Main Section~\ref{sec:med-quadratic} defines the primary
diagonal-standardization metric.  Its endpoints change negligibly over the
reported ridge range.  By contrast, a full-covariance metric requires a
separate regularization choice and changes the upper endpoint of the
total-heterogeneity set.

\begin{center}\small
\begin{tabular}{lcc}
\toprule
Metric and its regularization & Radius & $V$ set ($\times10^{-4}$)\\
\midrule
Diagonal, ridge $1\times10^{-10}$ & $6.54$ & $[18.9,\,92.9]$\\
Full covariance, ridge $1\times10^{-6}$ & $5.24$ & $[29.5,\,122.4]$\\
Full covariance, ridge $1\times10^{-4}$ & $5.24$ & $[29.5,\,122.3]$\\
Full covariance, ridge $1\times10^{-2}$ & $5.07$ & $[29.6,\,117.8]$\\
Full covariance, ridge $1\times10^{-1}$ & $4.41$ & $[29.3,\,102.2]$\\
\bottomrule
\end{tabular}

\end{center}

The lower endpoint of total heterogeneity is positive under every reported
metric.  The diagonal metric remains primary because it was prespecified and
does not require an additional full-covariance regularization choice.

Dropping one replicate changes no target exposure by more than \RepLOOMax\
percent.  Repeated half-frame calculations change no exposure by more than
\RepHalfMax\ percent and preserve the ordering
\[
\mathfrak g_D(\ell_{\mathrm{mean}})
>
\mathfrak g_D(\ell_{\mathrm{pop}})
>
\mathfrak g_D(\ell_{75-25})
>
\mathfrak g_D(\ell_{\mathrm{slope}}).
\]

\subsection{Official replicate calculations and primary numerical law}

Let \(\widehat\tau\) be the full-sample effect vector and
\(\widehat\tau^{(r)}\), \(r=1,\ldots,80\), the effect vectors obtained after
recomputing the state--year rates and re-solving the entire first stage with
official replicate \(r\).  The regular benchmark covariance is
\[
\widehat\Sigma_{\mathrm{SDR}}^\tau
=
\frac4{80}
\sum_{r=1}^{80}
(\widehat\tau^{(r)}-\widehat\tau)
(\widehat\tau^{(r)}-\widehat\tau)^\top.
\]
It is not the primary law when active-face regularity is uncertain.  The
primary primitive factor uses the corresponding full-estimate-centered
state--year replicate differences and generates the numerical directional
effect draws described in Supplementary Appendix~A.

The holdout ratios are used only to anchor a deterministic sensitivity scale.
No probability statement transporting the validation ratios to the
post-treatment mismatch is made in this paper.

\subsection{Optimization and replicate diagnostics}
\label{app:optimization-diagnostics}

\begin{table}[htbp]
\centering\small
\caption{Hard-simplex and ACS replicate diagnostics for the primary design.}
\label{tab:active-face-diagnostics}
\begin{tabular}{lcc}
\toprule
Diagnostic & Donor programs & Time programs\\
\midrule
Number of fitted programs & \ActNDonor & \ActNTime\\
Programs with an active-set change across \ActReplicates\ replicates
 & \ActDonorChanges & \ActTimeChanges\\
Minimum positive weight & \ActMinPositiveDonor & \ActMinPositiveTime\\
Minimum inactive KKT multiplier & \ActMinInactiveDonor & \ActMinInactiveTime\\
Maximum KKT residual & \ActMaxKKTDonor & \ActMaxKKTTime\\
Maximum reduced-KKT condition number & \ActCondDonor & \ActCondTime\\
Optimization failures & \ActFailDonor & \ActFailTime\\
\midrule
\multicolumn{3}{l}{Covariance diagnostics}\\
Rank / effective rank, centered SDR covariance
 & \multicolumn{2}{c}{\SDRRank\ / \SDREffectiveRank}\\
Smallest eigenvalue before PSD regularization
 & \multicolumn{2}{c}{\SDRMinEig}\\
Target-level change from PSD regularization
 & \multicolumn{2}{c}{\SDRPSDTargetChange}\\
\bottomrule
\end{tabular}
\tablenote{The table records numerical diagnostics under the official ACS
replicate perturbations.  Active-set changes do not select another inferential
branch: every primary interval, region, and test uses the numerical directional
law.  The exact critical-cone calculation checks the same fitted map.}
\end{table}

\subsection{Cross-year replicate coupling}\label{app:coupling}
The within-year SDR identity determines the role of replicate $r$ within an
annual file.  It does not determine whether replicate $r$ in 2009 and
replicate $r$ in 2014 should receive the same Gaussian coefficient.  In the
primary construction, we assign one coefficient to each replicate index and
use it across years.  This common-index coupling is the design assumption in
Assumption~\ref{as:acs-tangent-coupling}, not a consequence of the annual SDR
identity.

As a sensitivity, year-block coupling assigns independent coefficients across
years and removes cross-year dependence from the primitive direction.  Using
the same seed and Gaussian variates, we recompute every reported quantity under
this alternative.  The simultaneous interval widths are
\CouplingWidthMin--\CouplingWidthMax\ times their primary values.  The
critical radius is \CouplingRadiusBlock\ rather than \CouplingRadiusCommon;
the corresponding $V$ set is \CouplingVSetBlock\ rather than
\MedVDirectionalSet, and the $R^2$ set is \CouplingRtwoSetBlock\ rather than
\MedRtwoSet.  The boundary value remains \CouplingBoundaryBlock.

The signs and target ordering are unchanged, but interval widths and
quadratic endpoints are not identical.  We use common-index coupling as the
prespecified primary construction and report year-block coupling as a
sensitivity.  The annual SDR identity itself does not select either cross-year
joint law.

\subsection{Sensitivity to the ACS replicate frame}\label{app:replicate-sensitivity}
The donor exposures use the covariance constructed from the
\RepFrameSize\ released replicate columns.  Leave-one-replicate calculations
change no exposure by more than \RepLOOMax\ per cent.  Across the reported
half-frame calculations, the maximum change is \RepHalfMax\ per cent, and the
ordering of the four target exposures is unchanged.  These checks concern the
finite target-relevant functionals in
Assumption~\ref{as:replicate-tangent-supp}; they do not estimate an unrestricted
survey covariance in the full primitive dimension.

\section{Fixed-Donor Concentration Diagnostic}
\label{app:fixed-donor}

Theorem~\ref{thm:donor-regimes} distinguishes a null donor loading
from a fixed positive loading.  A rank-four Gaussian experiment illustrates
the two concentration patterns.

The simulated target error is
\begin{equation}
\widehat\theta_{G_1}-\theta_{G_1}
=
G_1^{-1/2}Z_u-\ell^\top V_D,
\qquad
Z_u\sim N(0,1),
\qquad
V_D\sim N(0,\Omega_D),
\label{seq:fd-diagnostic}
\end{equation}
with independent own and donor components.  Hence
\begin{equation}
\operatorname{Var}
(\widehat\theta_{G_1}-\theta_{G_1})
=
G_1^{-1}+\ell^\top\Omega_D\ell.
\label{seq:fd-diagnostic-var}
\end{equation}
For a donor-null direction,
\(\ell^\top\Omega_D\ell=0\), so the dispersion contracts at the own-noise
rate.  For a fixed positive-loading direction, the dispersion converges to
\(\{\ell^\top\Omega_D\ell\}^{1/2}\).  The reported interval uses the full
variance in~\eqref{seq:fd-diagnostic-var}; nominal coverage is therefore
compatible with a nonvanishing uncertainty floor.
Table~\ref{tab:fd-check} displays the resulting contraction or positive
limiting floor for the four fixed directions.

\begin{table}[t]
\centering\footnotesize
\caption{Concentration by target loading in the rank-four fixed-donor
experiment.}
\label{tab:fd-check}
\begin{tabular}{lrrrrr}
\toprule
Direction & Floor & SD, $G_1=50$ & SD, $G_1=100$ & SD, $G_1=200$ & 95\% coverage\\
\midrule
Donor-null eigenvector & $0.000$ & $0.141$ & $0.101$ & $0.069$ & $[0.947,\,0.959]$\\
Smallest-positive eigenvector & $0.283$ & $0.318$ & $0.299$ & $0.295$ & $[0.949,\,0.951]$\\
Fitted centered-slope direction & $0.372$ & $0.394$ & $0.383$ & $0.382$ & $[0.946,\,0.952]$\\
Leading eigenvector & $1.000$ & $1.015$ & $1.012$ & $0.992$ & $[0.943,\,0.952]$\\
\bottomrule
\end{tabular}

\tablenote{The donor-shock dimension is fixed and each cell uses 3,000 Monte
Carlo replications.  The donor floor is
\(\{\ell^\top\Omega_D\ell\}^{1/2}\).  The own-noise standard deviations are
\(1/\sqrt{50}=0.141\), \(1/\sqrt{100}=0.100\), and
\(1/\sqrt{200}=0.071\).  The null-space direction has zero donor floor and
contracts with \(G_1\).  The other directions have fixed positive loadings and
approach their corresponding floors.  The final column gives the range of
coverage of the full donor-plus-own Gaussian interval over
\(G_1\in\{50,100,200\}\).}
\end{table}

Coverage of the full-variance interval is nominal in each direction.  Its
width contracts only for the null-space direction; a fixed positive loading
leaves a nonvanishing donor floor.

The root-rate regime in Main
Theorem~\ref{thm:donor-regimes}(iii) is not simulated in this
compact table: the positive-loading directions here are fixed as \(G_1\)
changes.  The table is intended to illustrate the contrast between
null-space concentration and a nonvanishing fixed-donor variance floor.  Theorem~\ref{thm:donor-regimes}(vi) gives the corresponding
quadratic condition; no quadratic Monte Carlo claim is made in this appendix.

\bigskip
\noindent\textbf{Use of generative AI.}\ During manuscript preparation, the
authors used OpenAI ChatGPT (GPT-5.5) to check mathematical derivations and
proofs and to proofread the English text, and Anthropic Claude (Opus 4.8) to
assist with drafting and debugging the analysis and simulation code. Records of
the tool use and the verification steps are retained by the authors. The authors
take full responsibility for the mathematical statements, the proofs, the code
and the reported numbers.

\begin{spacing}{1}

\end{spacing}
\end{document}